\documentclass{elsarticle}
\usepackage{amsfonts,amssymb,amsmath,latexsym,ae,aecompl,bbm,algorithm}
\usepackage{algorithm}
\usepackage[noend]{algpseudocode}
\usepackage{graphics}
\usepackage{graphicx}
\usepackage{color}
\usepackage{amsthm,enumitem}
\usepackage{fullpage}

\usepackage{url}

\algblockdefx[ExecBl]{BlockOn}{BlockOff}  [1]{#1}
  
\makeatletter
\ifthenelse{\equal{\ALG@noend}{t}}%
  {\algtext*{BlockOff}}
  {}%
\makeatother

\newcommand{\remove}[1]{}
\newtheorem{fact}{Fact}

\usepackage{amsthm}

\makeatletter
\newtheorem*{rep@theorem}{\rep@title}
\newcommand{\newreptheorem}[2]{%
\newenvironment{rep#1}[1]{%
 \def\rep@title{#2 \ref{##1}}%
 \begin{rep@theorem}}%
 {\end{rep@theorem}}}
\makeatother

\newtheorem{theorem}{Theorem}
\newreptheorem{theorem}{Theorem}

\newtheorem{lemma}{Lemma}
\newreptheorem{lemma}{Lemma}

\newtheorem{claim}{Claim}

\newcommand{\INDState}[1][1]{\State\hspace{4ex}}

\usepackage{xparse}%

\usepackage{blindtext}

\usepackage{morewrites}

\usepackage{wrapfig}
\usepackage{xkeyval}%

\newwrite\authorbibfile%

\AtBeginDocument{%
	\immediate\openout\authorbibfile=\jobname.aub%
}%

\AtEndDocument{%
	\immediate\closeout\authorbibfile
	\InputIfFileExists{\jobname.aub}{}{}
}%

\makeatletter

\define@key{authorbib}{scale}[1]{%
	\def\AuthorbibKVMacroScale{#1}%
}

\define@key{authorbib}{wraplines}[10]{%
	\def\AuthorbibKVMacroWraplines{#1}%
}

\define@key{authorbib}{imagewidth}[4cm]{%
	\def\AuthorbibKVMacroImagewidth{#1}%
}

\define@key{authorbib}{overhang}[10pt]{%
	\def\AuthorbibKVMacroOverhang{#1}%
}

\define@key{authorbib}{imagepos}[l]{%
	\def\AuthorbibKVMacroImagepos{#1}%
}

\makeatother

\presetkeys{authorbib}{imagepos=l,imagewidth=4cm,wraplines=15,overhang=20pt}{}

\newlength{\AuthorbibTopSkip}
\newlength{\AuthorbibBottomSkip}
\NewDocumentCommand{\authorbibliography}{+o+m+m+m}{%
	\IfNoValueTF{#1}{%
	}{%
	\setkeys{authorbib}{#1}%
	\immediate\write\authorbibfile{%
		\string\begin{wrapfigure}[\AuthorbibKVMacroWraplines]{\AuthorbibKVMacroImagepos}[\AuthorbibKVMacroOverhang]{\AuthorbibKVMacroImagewidth}^^J
			\string\includegraphics[scale=\AuthorbibKVMacroScale]{#2}^^J
			\string\end{wrapfigure}^^J
	}%
}%
\IfNoValueTF{#3}{%
	\typeout{Warning: No author name}%
}{%
\immediate\write\authorbibfile{%
	\unexpanded{\vspace{\AuthorbibTopSkip}}^^J
	\string\noindent\relax
	\unexpanded{\textbf{#3}\par}^^J
	\string\noindent\relax
	\unexpanded{#4}^^J%
	\unexpanded{\vspace{\AuthorbibBottomSkip}}^^J
}%
}%
}%

\makeatletter
\def\ps@pprintTitle{%
 \let\@oddhead\@empty
 \let\@evenhead\@empty
 \def\@oddfoot{}%
 \let\@evenfoot\@oddfoot}
\makeatother

\begin{document}

\title{Deterministic Protocols in the SINR Model without Knowledge of Coordinates}

\author[iitm]{William~K.~Moses~Jr.\fnref{fn1}}
\ead{wkmjr3@gmail.com}

\author[indp]{Shailesh~Vaya\fnref{fn2}}
\ead{shailesh.vaya@gmail.com}

\fntext[fn1]{Corresponding author. This work was done while the author was an intern at Xerox Research Centre India, Bangalore. \textbf{Present Address:} The author is currently a postdoctoral fellow with the Faculty of Industrial Engineering and Management, Technion - Israel Institute of Technology, Haifa, Israel.}
\fntext[fn2]{\textbf{Present Address:} The author is currently an independent consultant in Bangalore, India. This work was done while the author was with Xerox Research Centre India, Bangalore.}

\address[iitm]{Department of Computer Science and Engineering, Indian Institute of Technology Madras, Chennai, India, 600036}

\address[indp]{Xerox Research Centre India, Bangalore, India, 560048}


\begin{abstract}
Much work has been developed for studying the classical broadcasting problem in the SINR (Signal-to-Interference-plus-Noise-Ratio) model for wireless device transmission. The setting typically studied is when all radio nodes transmit a signal of the same strength. This work studies the challenging problem of devising a distributed algorithm for multi-broadcasting, assuming a subset of nodes are initially awake, for the SINR model when each device only has access to knowledge about the total number of nodes in the network $n$, the range from which each node's label is taken $\lbrace 1,\dots,N \rbrace$, and the label of the device itself. Specifically, we assume no knowledge of the physical coordinates of devices and also no knowledge of the neighborhood of each node.
	
We present a deterministic protocol for this problem in $O(n \lg N \lg n)$ rounds. There is no known polynomial time deterministic algorithm in literature for this setting, and it remains the principle open problem in this domain. A lower bound of $\Omega(n \lg N)$ rounds is known for deterministic broadcasting without local knowledge~\cite{JKS-ICALP-13}. 
	
In addition to the above result, we present algorithms to achieve multi-broadcast in $O(n \lg N)$ rounds and create a backbone in $O(n \lg N)$ rounds, assuming that all nodes are initially awake. For a given backbone, messages can be exchanged between every pair of connected nodes in the backbone in $O(\lg N)$ rounds and between any node and its designated contact node in the backbone in $O(\Delta \lg N)$ rounds.
\end{abstract}

\maketitle


\noindent
\textbf{Keywords:} Distributed algorithms, 
Multi-Broadcast, 
Non-Spontaneous Wakeup, 
Backbone Creation, 
Signal-to-Interference-plus-Noise-Ratio model, 
Wireless networks, 
Deterministic algorithms, 
Strongly Selective Family based Dilution

\section{Introduction}\label{intro}

The SINR (Signal-to-Interference-plus-Noise-Ratio) model for communication for ad-hoc wireless networks has been deeply studied in recent years. It captures the effects of both (i) the strength of devices and (ii) the distance between devices on a message reaching its intended target. The protocols developed in this rich area have been complex and are either randomized or use the physical coordinates of the nodes and/or the knowledge of nodes' neighbors' labels. 

However, in the real world, it is not practical to assume that we know the physical coordinates of the nodes after they have been deployed. Similarly, the assumption that a node is aware of its neighborhood is also quite restrictive. Finally, randomized algorithms, while good, only work with high probability. Thus, from a very pragmatic deployment perspective, the following properties are desirable for communication protocols developed for ad-hoc wireless networks: (a) They are deterministic. (b) They support uncoordinated wakeup of nodes.\footnotemark[5]\footnotetext[5]{When all nodes are awake initially, it's called a \textbf{spontaneous wakeup}. When a nonempty subset of all nodes is awake initially (and other nodes are passive until they receive a message), it's called an \textbf{uncoordinated wakeup}.} (c) They assume availability of minimal knowledge like node's own label, maximum number of nodes $n$, label range $N$. (d) They have small round complexity and small communication complexity.

  In this work, we develop communication protocols for this setting, assuming the \textbf{weak connectivity} for SINR model formulated by Daum et al.~\cite{DGKN13} and further refined by Jurdzi\'nski and Kowalski~\cite{JK16} and \textbf{weak devices}, as formulated by Jurdzi\'nski et al.~\cite{JKS-ICALP-13}. 

\subsection{Our Contributions}
  We present deterministic protocols for the following communication problems: (a) Wakeup: Starting from a state when an arbitrary non-empty subset of nodes is awake, wake up everyone. (b) Multi-broadcast: $k>0$ nodes are each given their own unique message to be transmitted to all the rest. (c) Backbone: Create a constant degree connected dominating set (CDS) with asymptotically the same diameter as the network and such that every node in the network is connected to $\geq 1$ nodes in the backbone.\footnotemark[6]\footnotetext[6]{Once a backbone network has been created, a variety of communication tasks like multi-broadcast, routing and so on can be solved on these networks using some simple distributed algorithms associated with the backbone network.} 
  
  Our main contribution is that we present the first polynomial time deterministic algorithm for the challenging open problem of multi-broadcast from an uncoordinated wakeup, in the weak connectivity, weak devices setting from \cite{JKS-ICALP-13}, that does not use the knowledge of nodes' coordinates or the knowledge of the labels of nodes' neighbors. It was shown in \cite{JKS-ICALP-13} that the lower bound to accomplish deterministic broadcast without local knowledge is $\Omega(n \lg N)$. The gap between upper and lower running time bounds for this problem and setting is now $O(\lg n)$.  

  For devising deterministic protocols for the setting in consideration, we extensively make use of the tool known as Strongly Selective Family (SSF) Based Dilution, introduced in \cite{MV16a}. This tool allows a node to successfully transmit a message to its neighbors, under certain conditions. It is the key to bypassing the usual requirement that nodes know their physical coordinates.

  Our protocols repeatedly use the following three sub-protocols: (a) \emph{Tree-Grower}, which creates a forest of trees from an arbitrary uncoordinated start state. (b) \emph{Tree-Cutter}, which cuts down the trees in a forest to trees of height at most one, called \textbf{stars}. (c) \emph{Token-Passing-Transfer}, which allows a forest of awake nodes to transmit messages to each other and surrounding nodes. We use these sub-protocols, assuming different upper bounds on the number of participating nodes, in a time-multiplexed fashion to construct Algorithm \emph{Wakeup}, which runs in $O(n \lg N \lg n)$ rounds. In order to create the backbone, we develop \emph{Backbone-Creation} which uses \emph{Tree-Grower} and \emph{Tree-Cutter}, and then runs in $4$ stages to create a backbone in $O(n \lg N)$ rounds assuming spontaneous wakeup. For uncoordinated wakeup, we run \emph{Wakeup} and then \emph{Backbone-Creation} in $O(n \lg N \lg n)$ rounds. Finally, we develop \emph{Backbone-Message-Exchange} for nodes within the backbone to communicate with each other in $O(\lg N)$ rounds, and \emph{Backbone-Message-Transmit} for nodes outside the backbone to communicate with their contact node within the backbone in $O(\Delta \lg N)$ rounds. Our \emph{Multi-Broadcast} protocol uses \emph{Backbone-Creation} to create a backbone and subsequently employs \emph{Backbone-Message-Transmit} and \emph{Backbone-Message-Exchange} to transmit all messages to all nodes. It takes $O(n \lg N)$ rounds, assuming spontaneous wakeup. For uncoordinated wakeup, we first run \emph{Wakeup} and then run \emph{Multi-Broadcast} for a running time of $O(n \lg N \lg n)$ rounds.  In order to achieve our results, we require message size to be $O(\Delta \lg N)$ bits\footnotemark[8]\footnotetext[8]{We need messages of $O(\Delta \lg N)$ bits size for \emph{Tree-Grower}.}.

  Our protocols extensively and repeatedly employ the tool of SSF based dilution. Coincidentally and quite interestingly, the structure of the main algorithms for the wakeup and multi-broadcasting problems have a broader structure akin to deterministic algorithms developed for broadcasting and gossiping developed in the literature for ad-hoc radio networks \cite{CGR00,GRX04,KP04c,Gasieniec10}. Thus, we feel that our work serves an important role in connecting the very rich older literature on ad-hoc radio networks and the newer literature on SINR model via the use of new available tools. 
  
  Slides to aid in understanding our results can be found at \cite{TalkSlides17}. 
  Note that these slides were made for an earlier version of the paper that presented both the deterministic algorithms as well as the technique of SSF Based Dilution from \cite{MV16a}. Additionally, results are improved in the current version of the paper.

\subsection{Related Results}
\label{previousresults}
  The SINR model was introduced by Gupta and Kumar~\cite{GK00}, and various problems in this model have been formulated and considered in \cite{ALPP09,FKRV09,KV10,K11}. Recently, a division has been introduced in the SINR model between strong devices versus weak devices \cite{JKS-ICALP-13,KV10}. Strong devices only need to have their SINR ratio exceed the threshold for their message to be heard by a node, but weak devices also have to satisfy a second inequality which further limits their range. Another distinction appears when nodes have knowledge of their coordinates on the Euclidean plane and when they do not. Jurdzi\'nski et al.~\cite{JKRS13} achieve randomized broadcast with strong devices which have knowledge of their coordinates and uncoordinated wakeup in $O((D + \lg (\frac{1}{\zeta}))\lg n)$ rounds with probability at least $1 - \zeta$. In each of \cite{DGKN13,JKRS14,JKRS15,JKRS18}, the authors consider broadcast in the case of strong devices without knowledge of their coordinates. Daum et al.~\cite{DGKN13} achieve randomized broadcast with uncoordinated wakeup in $O(D \lg n \lg^{\alpha + O(1)} (R_s))$ rounds with high probability, where $\alpha$ is the path loss parameter of the SINR model, $D$ is the diameter of the communication graph, and $R_s$ is the ratio between the maximum distance between two communicable nodes to the minimum distance between two communicable nodes in the communication graph. In the case of Jurdzi\'nski et al.~\cite{JKRS14}, when they consider spontaneous wakeup, they achieve randomized broadcast in $O(D \lg n + \lg^2 n)$ rounds with high probability. When they consider uncoordinated wakeup, they achieve randomized broadcast in $O(D \lg^2 n)$ rounds with high probability. In Jurdzi\'nski et al.~\cite{JKRS15}, they achieve randomized multi-broadcast in $O(D \lg^2 n + k \lg n + lg^2 n)$ rounds with high probability, where $k$ messages stored at $k$ nodes need to be broadcasted. Jurdzi\'nski et al.~\cite{JKRS18} achieve deterministic broadcast in $O(D(\Delta + \lg^* N)\lg N)$ rounds, where $\Delta$ is the maximum degree of any node in the graph.

  Another division has been introduced by Daum et al.~\cite{DGKN13} and refined by Jurdzi\'nski and Kowalski~\cite{JK16} in the form of weak links versus strong links in the communication graph. A communication graph built on weak links has edges between nodes within range of each other. A strong link graph has edges between nodes which are within a constant fraction of the range of each other, where range is determined by strength of the device. Note that Daum et al.~\cite{DGKN13} define weak links assuming the presence of strong devices, whereas Jurdzi\'nski and Kowalski~\cite{JK16} differentiate between the strength of devices and strength of links of the communication graph. \cite{JKRS13,DGKN13,JKRS14,JKRS15,JKRS18} all have results for the strong links case. For the weak links case, Daum et al.~\cite{DGKN13} provide a randomized $O(n \lg^2 n)$ round algorithm with uncoordinated wakeup which uses strong devices and has no knowledge of coordinates and works with high probability. Jurdzi\'nski et al.~\cite{JRS17} improve this bound in the same setting with a $O(n \lg N)$ round deterministic algorithm and an $O(n \lg n)$ round randomized algorithm that works with high probability and does not require nodes to have labels. Jurdzi\'nski et al.~\cite{JKS-ICALP-13} use weak links and weak devices to design several deterministic algorithms. When nodes have knowledge of their coordinates and their neighbors, they present a $O(D \lg^2 N)$ round algorithm to achieve broadcast from an uncoordinated wakeup. When nodes know only their coordinates, they provide two algorithms with running times $O(n \lg N)$ rounds and $O(D \Delta \lg^2 N)$ rounds.  Furthermore, they prove a lower bound of $\Omega(n \lg N)$ rounds for deterministic broadcasting with uniform transmission powers and without local knowledge of immediate neighborhood. Chlebus and Vaya~\cite{CV16}, which is a preliminary version of Chlebus et al.~\cite{CKV15}, use weak links and weak devices with no knowledge of their coordinates to achieve randomized broadcast with uncoordinated wakeup in $O(n \lg^2 N)$ rounds with high probability. Reddy et al.~\cite{RKV15,RV16} use weak devices and weak links to construct deterministic multi-broadcast algorithms with uncoordinated wakeup in the presence of different types of knowledge. When nodes know their own coordinates and the coordinates of their neighbors, they achieve multi-broadcast in $O(D \lg^2 n + k \lg \Delta)$ rounds. When nodes know their own coordinates but have no information of their neighbors, they achieve multi-broadcast in $O((n+k)\lg n)$ rounds. When they do not know their coordinates but know the labels of their neighbors, they achieve multi-broadcast in $O((n+k)\lg n)$ rounds. We summarize relevant previous results in Table~\ref{tab:alg-comparison}. 

  Yu et al.~\cite{YHWYL13} achieve randomized multi-broadcast in an asynchronous system in the case of weak devices and weak links without knowledge of their coordinates in $O((D+k) \lg n + \lg^2 n)$ time-slots. Yu et al.~\cite{YHWTL16} achieve randomized multi-broadcast in the case of weak devices and weak links without knowledge of their coordinates in $O(D + k + \lg^2 n)$ rounds assuming spontaneous wakeup. However, they allow each station to control its power.
  
  Work has been done on dominating sets in the SINR model by Scheideler et al.~\cite{SRS08}. They present an algorithm to find a dominating set that stabilizes in $O(\lg n)$ rounds with high probability using tunable collision detection. Yu et al.~\cite{YHWTL16} use a randomized algorithm and power control to create a connected dominating set in $O(\lg^2 n)$ rounds with high probability. For a comprehensive survey on connected dominating sets in wireless ad hoc and sensor networks, we refer the reader to Yu et al.~\cite{YWWY13}. 
  
  Creating and communicating across a backbone in the SINR model was studied by Jurdzi\'nski and Kowalski~\cite{JK12} in the weak link, weak device model. They construct a backbone in $O(\Delta \lg^3 N)$ rounds. More work on the weak link, weak device model has been done in \cite{CKV15,CV16,RKV15,RV16,KMV16}. Chlebus et al.~\cite{CKV15,CV16} use randomization to create a backbone in $O(\Delta \lg^{7}N)$ rounds with high probability assuming all nodes are initially awake. For the setting when only some nodes are awake, they present a randomized algorithm to create a backbone in $O(n \lg^2 N + \Delta \lg^{7} N)$ rounds with high probability. Reddy et al.~\cite{RKV15,RV16} devise deterministic protocols to create a backbone when only some nodes are awake under different assumptions of knowledge. When nodes know their own coordinates and those of their neighbors, they create a backbone in $O(D \lg^2 n + k \lg \Delta)$ rounds. When nodes only know their own coordinates, backbone is created in $O(n \lg n)$ rounds. Their most interesting result is when nodes do not know their own coordinates, but know the labels of their neighbors. Somewhat surprisingly, they present a protocol that creates a backbone in $O(n)$ rounds. Kowalski et al.~\cite{KMV16} construct a backbone in $O(\Delta \lg^2 N)$ rounds when all nodes are initially awake and do not know their own coordinates but know their labels and the labels of their neighbors. Jurdzi\'nski et al.~\cite{JKRS15} consider the strong link, strong device model and create a quasi-backbone structure using randomization in $O(D \lg^2 n)$ rounds with high probability, where a quasi-backbone is the assignment of probabilities to nodes that allows groups of devices within certain distance of each other to communicate.

\begin{table}[ht]
\caption{Comparison of running times and other features of various algorithms to solve broadcast. $D$ is the diameter of the graph (based on strong/weak links), $\Delta$ is the max. degree of the graph, $n$ is the number of nodes, $N$ is the max. value of any label of a node, $\zeta$ is the maximal error probability, $\alpha$ is the path loss constant, $k$ is the number of nodes with messages to transmit, and $R_s$ is the maximum ratio between strong link lengths. All running times are in rounds. Running times of randomized algorithms are with high probability, except for RandUnknownBroadcast\cite{JKRS13}, in which case it is with probability at least $1 - \zeta$.}
\begin{center}
\resizebox{1.0\columnwidth}{!}{%
\begin{tabular}{|c|c|c|c|c|c|c|c|}
  \hline
  \textbf{\small Algorithm} & \textbf{\small Randomized} & \textbf{\small With}  & \textbf{\small Device} & \textbf{\small Link} & \textbf{\small Achieves} &\textbf{\small Running Time} & \textbf{\small Running Time}\\
  & & \textbf{\small Knowledge of} & \textbf{\small Type} & \textbf{\small Type} & \textbf{\small Multi-} & \textbf{\small with Spontaneous} & \textbf{\small with Uncoordinated} \\
  & & \textbf{\small Coordinates} & & & \textbf{\small Broadcast} & \textbf{\small Wakeup} & \textbf{\small Wakeup} \\
  \hline \hline
  RandUnknownBroadcast\cite{JKRS13} & Yes & Yes & Strong & Strong & No & - & $O((D + \lg (\frac{1}{\zeta}))\lg n)$ \\
  \hline
  StrongCast\cite{DGKN13} & Yes & No & Strong & Strong & No & - & $O(D \lg n \lg^{\alpha + O(1)} (R_s))$\\
  \hline
  SBroadcast, NoSBroadcast\cite{JKRS14} & Yes & No & Strong & Strong & No & $O(D \lg n + \lg^2 n)$ & $O(D \lg^2 n)$ \\
  \hline
  Broadcast\cite{JKRS15} & Yes & No & Strong & Strong & Yes & - & $O(D \lg^2 n + k \lg n + lg^2 n)$ \\
  \hline
  SMSBroadcast\cite{JKRS18} & No & No & Strong & Strong & No & - & $O(D(\Delta + \lg^* N)\lg N)$ \\
  \hline
  HarmonicCast\cite{DGKN13} & Yes & No & Strong & Weak & No & - & $O(n \lg^2 n)$ \\
  \hline
  BroadcastWithToken\cite{JRS17} & No & No & Strong & Weak & No & - & $O(n \lg^2 N)$ \\
  \hline
  BroadcastWithToken\cite{JRS17}* & Yes & No & Strong & Weak & No & - & $O(n \lg^2 n)$ \\
  \hline
  3-Timeslot Scheme\cite{YHWTL16}** & Yes & No & Weak & Weak & Yes & $O(D + k + \lg^2 n)$ & - \\
  \hline
  Modified\_DFS\cite{CKV15,CV16} & Yes & No & Weak & Weak & No & - & $O(n \lg^2 N)$ \\
  \hline
  DiamUBr\cite{JKS-ICALP-13}*** & No & Yes & Weak & Weak & No & - & $O(D \lg^2 N)$ \\
  \hline
  SizeUBr\cite{JKS-ICALP-13} & No & Yes  & Weak & Weak & No & - & $O(n \lg N)$ \\
  \hline
  GeneralBroadcast\cite{JKS-ICALP-13} & No & Yes & Weak & Weak & No & - & $O(D \Delta \lg^2 N)$ \\
  \hline
  Local-Multicast\cite{RKV15}**** & No & Yes & Weak & Weak & Yes & - & $O(D \lg^2 n + k \lg \Delta)$ \\
  \hline
  General-Multicast\cite{RKV15} & No & Yes & Weak & Weak & Yes & - & $O((n+k) \lg n)$ \\
  \hline
  BTD$\_$Traversals, BTD$\_$MB\cite{RKV15}***** & No & No & Weak & Weak & Yes & - & $O((n+k) \lg n)$ \\
  \hline
  \textbf{Multi-Broadcast, Wakeup} & No & No & Weak & Weak & Yes & $O(n \lg N)$ & $O(n \lg N \lg n)$ \\
  \hline
  \multicolumn{8}{|l|}{\small *BroadcastWithToken\cite{JRS17} is still deterministic, but the initial labels of nodes are taken randomly, leading to the better bound.}\\
  \multicolumn{8}{|l|}{\small **3-Timeslot Scheme\cite{YHWTL16} requires nodes to be able to control their power.}\\
  \multicolumn{8}{|l|}{\small ***DiamUBr\cite{JKS-ICALP-13} requires nodes to know coordinates of their neighbors.}\\
  \multicolumn{8}{|l|}{\small ****Local-Multicast\cite{RKV15} requires nodes to know coordinates of their neighbors.}\\
  \multicolumn{8}{|l|}{\small *****BTD$\_$Traversals, BTD$\_$MB\cite{RKV15} require nodes to know labels of their neighbors.}\\
  \hline
    
\end{tabular}
}
\end{center}
\label{tab:alg-comparison}
\end{table}

  The problem of local broadcasting,  which deals with transmitting a message to all its neighboring nodes, has been studied in \cite{GMW08,YWHL11,KV10,HM12,FW14}.
  A survey of approximation algorithms in the SINR model was performed by Goussevskaia et al.~\cite{GPW10}. Other sub-models within the SINR model have also been looked at recently by Jurdzi\'nski et al.~\cite{JKS-FCT-13}. There have been various trends in the SINR model with respect to signal strength and geometric decay \cite{AEKLPR12,DGKN13,BH14,JKRS14,CKV15,CV16}. There has also been an attempt to bridge the real world usefulness of the SINR model with the results available for the theoretically easier model of Unit Disc Graphs (UDG) in the form of Quasi-UDGs \cite{KWZ03,BFNO03}. Quasi-UDGs consider two nodes to be connected if the distance between them lies below a threshold $\gamma$, $0 < \gamma < 1$, disconnected if the distance is above $1$ and maybe connected if the distance lies between $\gamma$ and $1$. Some work has been done on converting results from one model to the other \cite{LL09}.

\subsection{Organization of the paper}
  The rest of this paper is organized as follows. Section~\ref{sect:prelims} introduces the SINR model as well as useful technical preliminaries. Section~\ref{sect:wakeup} presents the algorithm to achieve wakeup of nodes. Section~\ref{sect:backbone} presents the protocols to create and utilize a backbone subnetwork. Section~\ref{sect:multi-broadcast} develops the protocol for achieving multi-broadcast. We briefly present conclusions and an open problem in Section~\ref{sect:conclusions}.

\section{Preliminaries}\label{sect:prelims}

\paragraph{The SINR Model}
  Each wireless station, $u$, has some transmission power $P_u \in \mathcal{R}^+$. The Euclidean distance between two stations $u$ and $v$ is given by the distance function $d(u,v)$. For a given round, let $\mathcal{T}$ be the set of all stations which transmit in that round. The SINR for a station $u$'s message at station $v$ in that round is defined as follows:\\
$$SINR(u,v,\mathcal{T}) = \dfrac{\frac{P_u}{d(u,v)^{\alpha}}}{\mathcal{N} + \sum\limits_{i \in \mathcal{T} \setminus \{u\}} \frac{P_i}{d(i,v)^{\alpha}}}.$$
  $\alpha \geq 2$ and $\mathcal{N} \geq 0$ are fixed parameters of the model called the path loss constant and the ambient noise respectively. All the results in this paper hold only for $\alpha > 2$. A node $v$ receives $u$'s message iff the SINR ratio of $u$'s message at $v$ crosses a threshold $\beta \geq 1$, which is also a fixed parameter of the model: 
  \begin{align}
  \dfrac{\frac{P_u}{d(u,v)^{\alpha}}}{\mathcal{N} + \sum\limits_{i \in \mathcal{T} \setminus \{u\}} \frac{P_i}{d(i,v)^{\alpha}}} \geq \beta. \label{eq:sinr-ineq}
  \end{align}

A device which only needs to satisfy Inequality~\ref{eq:sinr-ineq} for its message to be heard is called a \textbf{strong device}. A \textbf{weak device} further needs to satisfy the following inequality, 
\begin{align}
\frac{P_u}{d(u,v)^{\alpha}} \geq (1 + \epsilon) \beta \mathcal{N}, \label{eq:weak-device-ineq}
\end{align}
where $\epsilon > 0$ is called the sensitivity parameter of the device. Note the relation between weak and strong devices. When $\epsilon = 0$, Inequality~\ref{eq:weak-device-ineq} reduces to Inequality~\ref{eq:sinr-ineq} (in the absence of interference). In this paper, we assume that all devices are weak devices and have the same fixed transmission power $P$.

We now describe what it means for devices to receive each other's message, using definitions commonly found in the literature, cf.~\cite{JKS-ICALP-13,KV10}. When both Inequality~\ref{eq:sinr-ineq} and~\ref{eq:weak-device-ineq} are satisfied in a given round, node $v$ \textbf{successfully receives} node $u$'s transmission in that round. Node $u$'s \textbf{transmission range} is the maximum distance at which another station can be located away from it and still successfully receive a message from $u$ when no others stations transmit. For the remainder of this paper, we use the terms `transmission range' and `range' interchangeably. Since all nodes have the same power, they also all have the same range, $r$. When all nodes within range of a node $u$ successfully receive its message in a given round, we say that $u$ \textbf{successfully transmitted} in that round. Note that only if a node was a receiver in that round will it actually receive $u$'s message and be counted in the definition. Otherwise, the message will be discarded.

  A \textbf{communication graph}, denoted by $G(V,E)$, is a graph where each station is considered a node and an edge from node $u$ to node $v$ denotes that $v$ is within range of $u$. Since all nodes have same range, the graph is undirected. We assume $G$ is connected. The \textbf{weak links} model \cite{DGKN13,JK16}, also known as weak connectivity, refers to how edges exist iff nodes are within range of each other as opposed to within a constant fraction of that range (known as the \textbf{strong link} model). We define a \textbf{star} as a subset of $G(V,E)$ forming a tree of height at most one.

  Any algorithm executed by a station proceeds in a series of rounds, where each round corresponds to one global clock tick. In a given round, a station may act either as a receiver or a transmitter, but not both. A transmitter may transmit a message of size $O(\Delta \lg N)$ bits, where $\Delta$ is the maximum degree of any node in the network. Thus, one round of an algorithm for a given node consists of the following three steps:
\begin{enumerate}
\item If the node acted as a receiver in the previous round, then it receives any message successfully sent to it in the last round.
\item The node performs some local computation, if any.
\item If the node is a transmitter in this round, then it transmits a single message.
\end{enumerate}

Nodes do not have the ability to detect collisions. 

Nodes which are asleep remain inactive until they receive a wakeup message from an awake node. Furthermore, these wakeup messages contain the current round number which allows all nodes to maintain synchronization.

\paragraph{Grid and Pivotal Grid}
  Consider a 2-dimensional grid $G_x$ overlayed on the Euclidean plane such that the length of each side of a grid box is $x$. Each grid box is denoted by the coordinates of its bottom left coordinates $(a,b)$. Therefore, a device with coordinates $(m,n)$ in the Euclidean plane will have \textbf{grid coordinates} $(a,b)$ on the grid $G_x$ iff $ax \leq m < (a+1)x$ and $bx \leq n < (b+1)x$. We represent the grid box of a station $u$ in grid $G_x$ as $C_{x,u}$.
  The \textbf{pivotal grid} is the grid $G_{\frac{r}{\sqrt{2}}}$. The significance of the pivotal grid is that any two nodes located within the same pivotal grid box are within range of each other. This has been useful in the design of various algorithms \cite{DP07,EGKPPS09}. The number of nodes located in a given box of the pivotal grid is not bounded.
  Grid boxes $x$ and $y$ are within range of each other iff there exist locations in $x$ and $y$ such that two nodes placed at these locations are within range of each other. \textbf{Box-distance} between two grid boxes with coordinates $(a_1, b_1)$ and $(a_2, b_2)$ is $0$ if the two boxes intersect, else $k$ where $k = max( min(|a_1 - a_2 - 1|, |a_2 - a_1 - 1|), min(|b_1 - b_2 - 1|, |b_2 - b_1 - 1|))$. Note that while we work with the imagination of a grid overlayed on the nodes, we never need to explicitly compute it for our algorithms.
  
\paragraph{Strongly Selective Family}
  Let $N \geq c$ and both $N$ and $c$ be positive integers. An ($N,c$)-strongly selective family, commonly abbreviated to ($N,c$)-ssf, is a family $F$ of subsets of integers from $[1,N]$ such that for any non-empty integer subset $S$ of $[1,N]$, $|S| \leq c$, for each element $x \in S$, there exists a set $R \in F$ such that $S \bigcap R = \{x\}$. There exist $F$'s of size $O(c^2 \lg N)$ which satisfy the above definition, cf. Clementi et al.~\cite{CMS01}. Recent work~\cite{BG17} shows that a possible explicit construction exists in the form of an $(N, (1,c-1))$ cover free family.
    
\paragraph{Strongly Selective Family Based Dilution}
 SSF based dilution is a technique that uses an ($N,c$)-ssf to allow neighbors of a node, within $\sqrt{2}x$ distance of it in a grid $G_x$, to successfully receive a message from it. Let $F$ be an ($N,c$)-ssf and let $F_1, F_2, \ldots, F_z$ be sets belonging to $F$. A node \textbf{$u$ executes an ($N,c$)-ssf} when $u$ is active (performs some action such as transmission) only in those rounds $i$ such that $u \in F_i$, and in other rounds $u$ just acts as a receiver. The size of $F$, $z = O(c^2 \lg N) = c_1 \lg N$. Here, $c = k^2(2d+1)^2$, where $d$ is a constant and comes from Lemma~\ref{lem:ssf-dil}, restated from \cite{MV16a} below.
 
\begin{lemma}{[Lemma~2 in \cite{MV16a}]}\label{lem:ssf-dil}
	For stations with same range $r$, sensitivity $\epsilon > 0$, and transmission power, for each $\alpha > 2$, there exists a constant $d$, which depends only on the parameters $\alpha$, $\beta$, and $\epsilon$ of the model and a constant $k$, satisfying the following property. 
	
	Consider a set of stations $W$. Let $G_x$, where $x \leq \frac{r}{\sqrt{2}}$, be a grid such that $\min_{i,j \in W, C_{x,i} \neq C_{x,j}}d(i,j) = \sqrt{2} x$. Let $u,v \in W$ be two stations such that $C_{x,u} \neq C_{x,v}$ and $d(u,v) = \sqrt{2} x$. Let $A_u$ be the set of stations in $C_{x,u}$. If for every grid box of $G_x$, at most a constant $k$ stations of $W$ want to transmit in any given round, then the following property holds.
	
	If $u$ is transmitting in a round $t$ and no other station within its box or a box less than $d$ box distance away from its box is transmitting in that round, then $v$ and all stations in $A_u$ can hear the message from $u$ in round $t$.
\end{lemma}
 
  For a given grid $G_x$, $x \leq \frac{r}{\sqrt{2}}$, $k$ is usually an upper limit on the number of active nodes present in any box of $G_x$. Because of the way we use ($N,c$)-ssf's throughout this paper, it suffices to set $k=1000$. Due to this $c$ is a constant and by extension $c_1$ is also a constant and the number of rounds of an ($N,c$)-ssf execution is $O(\lg N)$. Furthermore, by the following theorem, restated from \cite{MV16a} below, we are guaranteed that executing an ($N,c$)-ssf under the required conditions allows all neighbors of a node at most $\sqrt{2}x$ away from it in $G_x$ to successfully receive its message.
  
   \begin{theorem}{[Theorem~1 in \cite{MV16a}]}\label{the:ssf-replace}
   	For a grid $G_x$, $x \leq r/\sqrt{2}$, let set of all nodes that want to transmit satisfy properties of Lemma~\ref{lem:ssf-dil}. Every node in this set can successfully transmit a message to its neighbors within $\sqrt{2}x$ distance of it in $O(\lg N)$ rounds by executing one $(N,c)$-ssf, where $c = k^2(2d+1)^2$ ($d$~is the constant that bounds box distance away from node within which other nodes must be silenced, taken from Lemma~\ref{lem:ssf-dil}. $k$ is the constant that bounds number of nodes from the set in any box of the grid).
   \end{theorem}

\paragraph{Knowledge of Stations}
 Nodes know the value of $n$, $N$, their own label, and a common ($N,c$)-ssf schedule. Nodes do not know the value of their coordinates or labels of their neighbors.

\paragraph{Problem Statements}
 Using the above tools and techniques, in the model described, we attempt to solve the following problems. Our goal is to solve each of these problems in the minimum number of rounds.

 \textbf{Wakeup:} Initially, $k$ nodes, $1 \leq k < n$, are awake and can transmit messages. An asleep node can be woken up if it hears a message from an awake node. Our goal is to wake up all nodes in the network.

 \textbf{Multi-Broadcast:} Initially, $k$ nodes, $1 \leq k \leq n$, each have a unique message. Our goal is to transmit these messages so that all nodes have all $k$ messages.

 \textbf{Backbone:} We need to create a backbone and subsequently design protocols for nodes to communicate within the backbone. A backbone is essentially an overlay network which facilitates speedy transfer of messages between nodes. More formally, according to Jurdzi\'nski and Kowalski~\cite{JK12}, the properties that need to be satisfied for a network to be considered a backbone are:
 \begin{enumerate}
 \item The nodes of the backbone, $H$, form a connected dominating set which induces a subgraph of the communication graph, $G$, and have a constant degree relative to other nodes within the backbone.
 \item The number of nodes in $H$ is $O(s\_c\_d)$, where $s\_c\_d$ is the size of the smallest connected dominating set of the communication graph.
 \item Each node $u$ of $G \backslash H$ is associated with exactly one node of $H$ which is also a neighbor of $u$, which acts as its entry point into the backbone.
 \item The asymptotic diameter of $H$ is same as that of $G$.
 \end{enumerate}
 Our first goal is to create a backbone that satisfies the above properties. Subsequently, we need to design a protocol to allow two nodes within the backbone to quickly exchange messages with each other. Finally, we need a protocol for nodes outside the backbone to communicate with their contact nodes in the backbone, also called their leaders.\footnote{Once a backbone is created, some nodes are in the backbone, and some are not. Each node $u$ not in the backbone is assigned a node $v$ which is both within the backbone and also a neighbor of $u$ in the communication graph, i.e. $u$'s contact node in the backbone. By sending and receiving messages from $v$, $u$ is able to utilize the backbone.}

\section{Wakeup}\label{sect:wakeup}

  In the wakeup problem, there are $k$ arbitrary nodes, $1 \leq k < n$, which are initially awake and the task is to wake up the remaining nodes. We develop the Algorithm \emph{Wakeup} to solve the problem and prove the following about it:

\begin{theorem}\label{the:wakeup-main}
  Algorithm \emph{Wakeup} successfully wakes up all nodes in $O(n \lg N \lg n)$ rounds, assuming initially at least one node is awake.
\end{theorem}

\subsection{Overview of Algorithm Wakeup}
\label{overview}
  Algorithm \emph{Wakeup} uses three main procedures: \emph{Tree-Grower}, \emph{Tree-Cutter}, and \emph{Token-Passing-Transfer}. We want every asleep node to receive a wake up message from an awake node. If we initially had every awake node try to transmit such a wake up message at the same time during an ($N,c$)-ssf schedule, there would be no guarantee that the number of transmitting nodes in any given pivotal grid box would be upper bounded by a constant. So we use \emph{Tree-Grower} and \emph{Tree-Cutter} to reduce the number of nodes which are allowed to transmit during a given ($N,c$)-ssf by connecting nodes together into trees and allowing at most one node from every tree to transmit during a given ($N,c$)-ssf. \emph{Tree-Grower} allows us to overlay a forest of trees on the network, but we cannot immediately use this forest to pass messages through the network because multiple trees might pass through the same pivotal grid box. In order to ensure that not too many nodes in the same pivotal grid try to transmit at the same time, we need to cut the size of trees down so that each tree is of height at most one using \emph{Tree-Cutter}. \emph{Token-Passing-Transfer} is used to actually transmit these wake up messages. Very briefly we discuss what each procedure achieves before we sketch Algorithm \emph{Wakeup}.\\

\noindent \textbf{Procedure Tree-Grower} The \emph{Tree-Grower} procedure, presented in Section~\ref{subsect:tg}, takes nodes which know nothing about each other and forms a forest of trees.

\begin{theorem}\label{the:tg-props}
  Assuming that the number of nodes given as input is an upper bound on the actual number of awake nodes in the network, the execution of Procedure \emph{Tree-Grower} on this set of nodes results in creating a forest in $O(n \lg N)$ rounds with the following properties:
\begin{enumerate}
\item Every node is either a leader or a child and belongs to exactly one tree.
\item Every tree has exactly one leader and there is at most one leader per grid box of the pivotal grid.
\item Every node knows the labels of its parent node and children.
\end{enumerate}
\end{theorem}

\noindent \textbf{Procedure Tree-Cutter} The \emph{Tree-Cutter} procedure, presented in Section~\ref{subsect:tc}, takes the trees previously formed by \emph{Tree-Grower} and creates stars, with the same guarantees as the \emph{Tree-Grower} procedure.
\begin{theorem}\label{the1}
  Assume that the nodes satisfy the properties described in Theorem~\ref{the:tg-props}. The execution of Procedure \emph{Tree-Cutter} takes $O(n \lg N)$ rounds and results in the creation of stars that satisfy the following properties:
\begin{enumerate}
\item Every node is either a leader or a follower and belongs to exactly one tree.
\item Each star has exactly one leader, which is the center of the star.
\item There exists at most one leader node per grid box of the pivotal grid.
\end{enumerate}
\end{theorem}

\noindent \textbf{Procedure Token-Passing-Transfer}
  Finally, \emph{Token-Passing-Transfer}, presented in Section~\ref{subsect:tpt}, is used to wake up all nodes which are within the communication range of the participating nodes. It assumes as input stars created by the \emph{Tree-Cutter} procedure and explores them with the help of a token. In particular, when a node receives a token, it transmits a wake up message and all asleep nodes that hear this message (formally) wake up.
\begin{theorem}\label{the:tpt-works}
	Assuming that the participating nodes are in a configuration characterized by Theorem~\ref{the1}, the \emph{Token-Passing-Transfer} procedure takes $O(n  \lg N)$ rounds to complete and ensures  that each participating node successfully transmits the message.
\end{theorem}

\paragraph{Wakeup}
  Given a set of awake nodes $A$, after executing the 3 procedures, we are able to wake up all neighboring nodes $B$. However, if there exist a set of asleep nodes $C$ which are neighbors of $B$ but not $A$, then we will have to execute the 3 procedures once again on the set of nodes $B$ in order to wake up these new neighbors.  If the communication graph is a line, then we would have to execute these 3 procedures $n-1$ times for a total running time of $O(n^2 \lg N)$. But notice that in each execution, we're unnecessarily considering that all $n$ nodes are participating, when only a fraction are. We use this key insight that the number of awake nodes participating in the 3 procedures may be a fraction of $n$ in order to develop \emph{Wakeup}, which runs in $O(n \lg N \lg n)$ rounds.

  \emph{Wakeup} consists of executing $O(n \lg N)$ phases. Each \textbf{phase} consists of executing $(\lfloor \lg n \rfloor + 1)$ slots. In turn, each \textbf{slot} consists of 3 rounds, which are respectively meant to execute a single round of Procedure \emph{Tree-Grower}, Procedure \emph{Tree-Cutter}, or Procedure \emph{Token-Passing-Transfer}. In a slot, a node participates in at most one of these rounds and does nothing in the other unutilized rounds. This achieves time-division-multiplexing of the executions of the three procedures. Figure~\ref{fig:one-phase-slot} shows this time-division-multiplexing.

  We give a top-down view of \emph{Wakeup}. Algorithm \emph{Wakeup} executes the three procedures for slot $i \in [1, \lfloor \lg n \rfloor + 1]$, with an estimate of $2^i$ on the upper bound of the number of participating nodes, $O(2^{\lfloor \lg n \rfloor + 1 - i})$ times. That is, these procedures for the smallest numbers of slots are executed $O(n)$ times, while these procedures for the largest (correct) estimate of $O(n)$ nodes are executed $O(1)$ times, because they take much longer. Thus, one phase of Algorithm \emph{Wakeup} comprises of execution of a single (next) round of (one of) the three procedures for each slot $i=\{1,\ldots, \lfloor \lg n \rfloor + 1\}$ and this is repeated till completion of the algorithm.  The complete execution of the three procedures for slot $i$ is called an \textbf{epoch} for that slot. The \textbf{$j^{\mbox{th}}$ epoch for slot $i$}, denoted as $e_i^j$, is the sequence of phases during which the three  procedures can be run for the $j+1^{\mbox{th}}$ time in slot $i$. Figures~\ref{fig:top-down-wakeup} and \ref{fig:epoch} graphically illustrate this\footnote{For figures appearing in this paper in print, there is no need to use color.}, with Figure~\ref{fig:epoch} giving extra insight into the timing of epochs.
 
\begin{figure}
\includegraphics[page=3,height=2.2in]{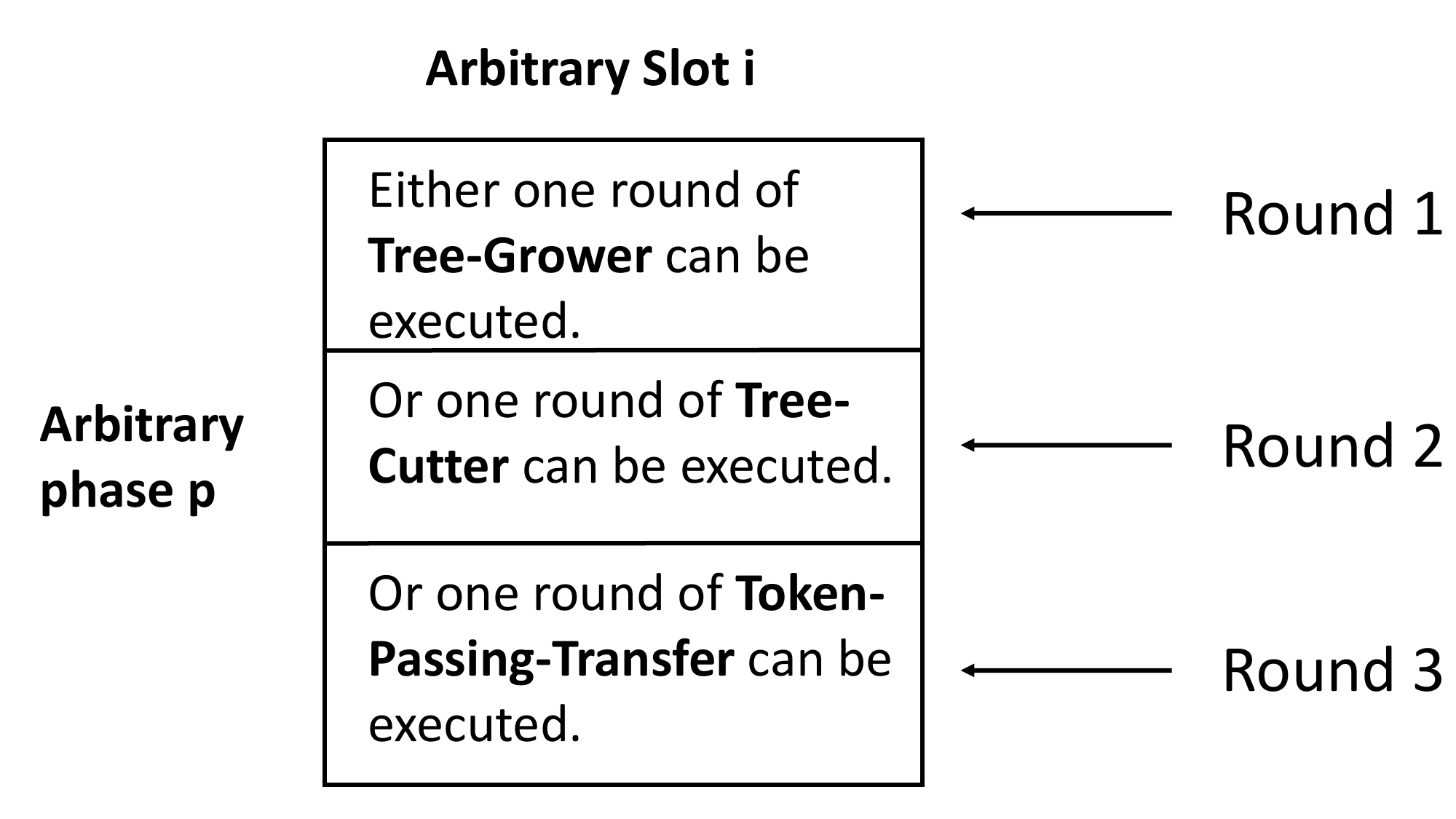}
\caption{A graphical illustration of what happens in one slot of one phase for a given node. For a given phase and slot, at most one round of \emph{Tree-Grower}, \emph{Tree-Cutter}, or \emph{Token-Passing-Transfer} will be executed depending on what stage the particular node is in in the given slot.} \label{fig:one-phase-slot}
\end{figure} 

\begin{figure}
\includegraphics[page=7,height=3.4in]{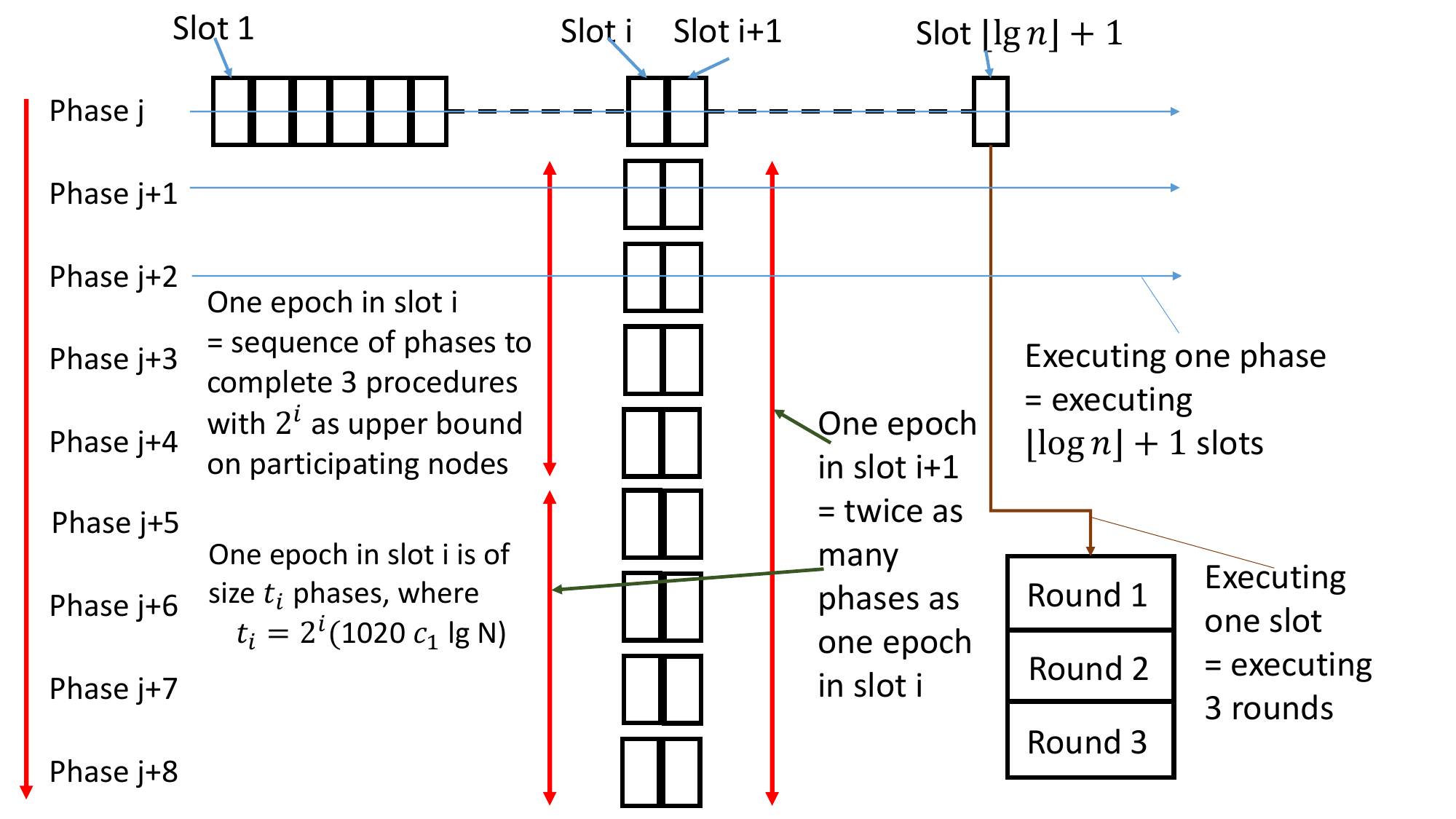}
\caption{Top-down view of \emph{Wakeup}.}\label{fig:top-down-wakeup} 
\end{figure}

\begin{figure}
\includegraphics[page=2,height=3in]{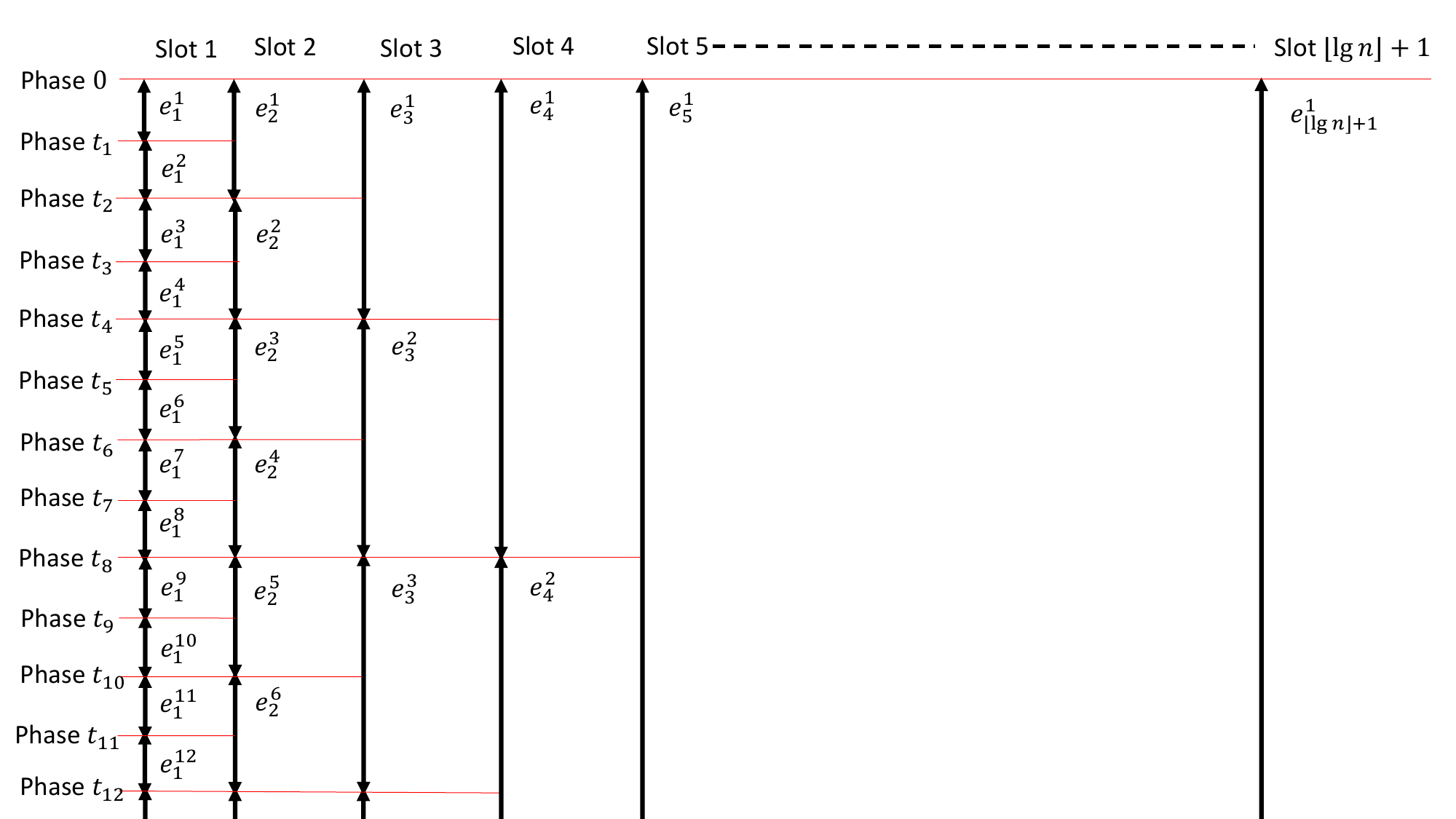} 
\caption{A graphical illustration of the lengths of various epochs and their relationships with each other. $e_i^j$ is the $j^{\text{th}}$ epoch for slot $i$. $t_i = 2^i \cdot 1020 c_1 \lg N$ 
	} \label{fig:epoch}
\end{figure}

\begin{figure}
\includegraphics[page=4,height=3.6in]{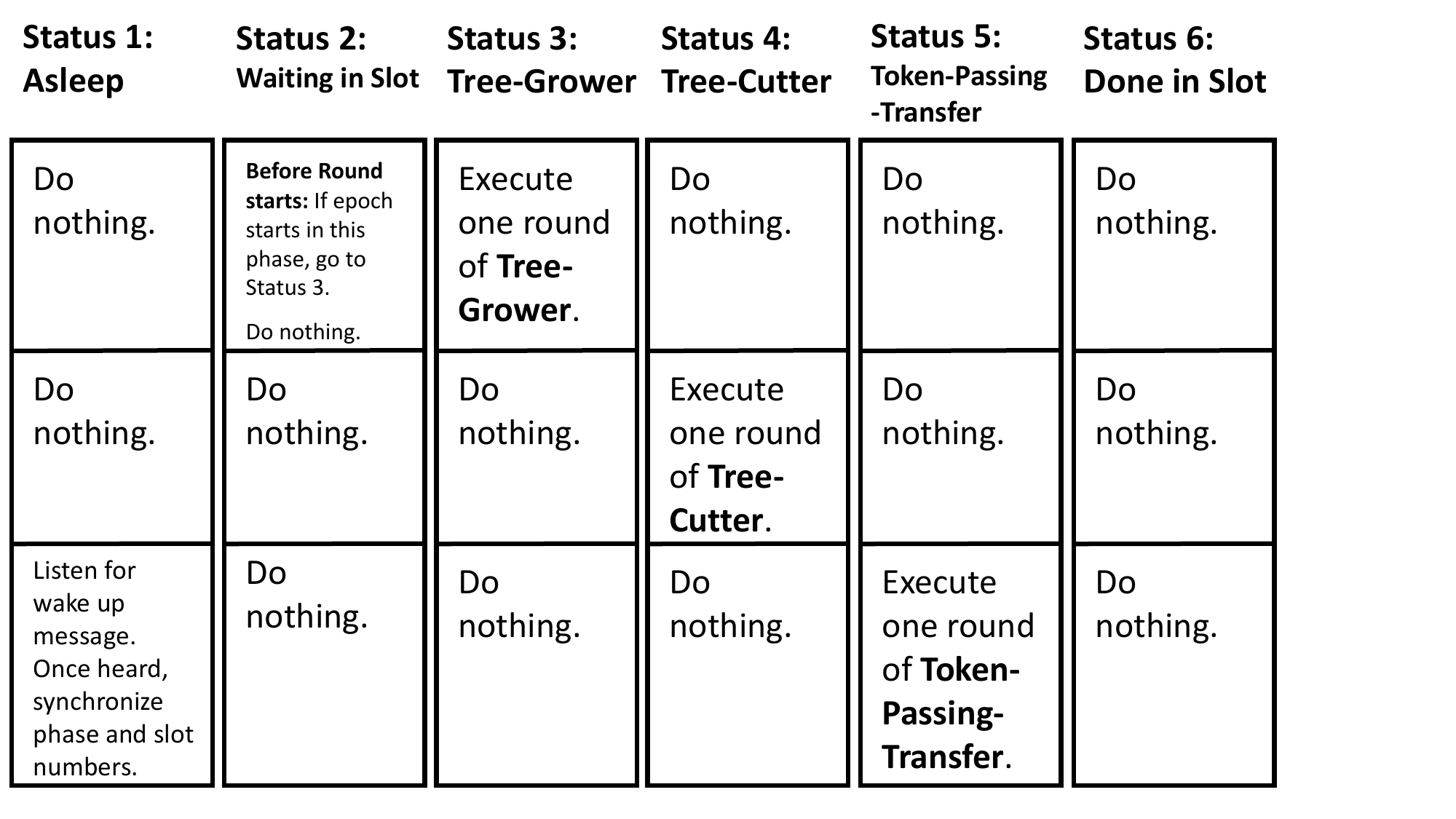} 
\caption{For a given phase and slot, a node will act in one of 6 ways depending on which stage it is in in that slot.} \label{fig:status-of-node-wakeup}
\end{figure}

 We now describe the execution of \emph{Wakeup} from the viewpoint of a single node. A node can be in one of 6 statuses, enumerated as 1. Asleep, 2. Waiting in Slot, 3. Tree-Grower, 4. Tree-Cutter, 5. Token-Passing-Transfer, and 6. Done in Slot. If a node is asleep initially (status 1), it needs to be woken up by a wakeup message transmitted during execution of Procedure \emph{Token-Passing-Transfer}; other messages are ignored by asleep nodes. A node once woken up in one slot, is awake in all slots. In addition to the wakeup message, the current phase and slot number are also transmitted for synchronization.  Once a node wakes up, it waits for the starting phase of a new epoch for each slot (status 2). Note that a node may be waiting in one slot and participating in an epoch in another slot. For a given slot, once the next epoch starts, the node participates in the 3 procedures (status 3-5), as described earlier, once and subsequently does nothing in that slot till the end of Algorithm \emph{Wakeup} (status 6). Every sequence of 3 rounds of the algorithm for a given node are thus referenced by phase $\#$, slot $\#$, and status $\#$. This is graphically illustrated in Figure~\ref{fig:status-of-node-wakeup}.


\alglanguage{pseudocode}
\begin{algorithm}
\caption{Wakeup(No. of nodes $n$), run by each node $u$}
\label{prot:multi-wakeup}
\begin{algorithmic}[1]
	\State $type$ denotes the type of node: $leader$, $follower$, or $neutral$. 
	\State $my\_leader$ denotes leader of a node (for a leader, $my\_leader$ is its own label). 
	\State Tree $T$ stores the children of a leader. For followers, $T$ is empty.  
	\State $slot$, $status[slot]$, and $t_i$ are explained in Figures~\ref{fig:top-down-wakeup} and \ref{fig:status-of-node-wakeup}.

	\State
	\State If awake set $status[slot] \gets 2$ else $status[slot] \gets 1$, $\forall slot \in [1, \lfloor \lg n \rfloor + 1]$

	\For {$phase \gets 0, 24n t_1 - 1$} 
		\For{ $slot \gets 1, \lfloor \lg n \rfloor + 1$}
			\BlockOn{ \textbf{Status 1 (Asleep):} $u$ is not awake}
				\State Do nothing for 3 rounds
				\If {$u$ hears a wakeup message}
					\State Wake up, synchronize $slot$ and $phase$, set $status[slot] \gets 2$, $\forall slot \in [1, \lfloor \lg n \rfloor + 1]$
				\EndIf
			\BlockOff
			
			\BlockOn{ \textbf{Status 2 (Waiting in Slot):} $status[slot] = 2$} 
				\If {$(phase \mod t_{slot}) = 0$} 
					\State $status[slot] \gets 3$ and goto Status 3.
				\EndIf
			
				\State Do nothing for 3 rounds

			\BlockOff
			
			\BlockOn{ \textbf{Status 3 (Tree-Grower):} $status[slot] = 3$}
				\State Execute next round of \emph{Tree-Grower($2^{slot}$)}
				\If {\emph{Tree-Grower} is done executing}
					\State Store the values returned in $\langle T[slot], type[slot], my\_leader[slot] \rangle $, set $status[slot] \gets 4$
				\EndIf
				\State Do nothing for 2 rounds
			\BlockOff
			
			\BlockOn{ \textbf{Status 4 (Tree-Cutter):} $status[slot] = 4$}
				\State Do nothing for 1 round
				\State \small Execute next round of \emph{Tree-Cutter($T[slot], type[slot], my\_leader[slot], 2^{slot}$)} \normalsize
				\If {\emph{Tree-Cutter} is done executing}
					\State Store the values returned in $\langle T[slot], type[slot], my\_leader[slot] \rangle $, set $status[slot] \gets 5$
				\EndIf
				\State Do nothing for 1 round
			\BlockOff
			
			\BlockOn{ \textbf{Status 5 (Token-Passing-Transfer):} $status[slot] = 5$}
				\State Do nothing for 2 rounds
				\State $msg \gets$ Wakeup message, $phase$, $slot$
				\State \small Execute next round of \emph{Token-Passing-Transfer($T[slot], type[slot], my\_leader[slot], 2^{slot}, msg$)} \normalsize
				\State \textbf{if} \emph{Token-Passing-Transfer} done executing \textbf{then} $status[slot] \gets 6$
			\BlockOff
			
			\BlockOn{ \textbf{Status 6 (Done in Slot):} $status[slot] = 6$}
				\State Do nothing for 3 rounds	
			\BlockOff
		\EndFor
	\EndFor

\Statex
\end{algorithmic}
\end{algorithm}

\subsubsection{Proof of Theorem~\ref{the:wakeup-main}}
 
  We will argue the correctness of \emph{Wakeup} as follows. First, we show new nodes will always be woken up and all nodes would eventually awaken should \emph{Wakeup} be run for an unbounded number of phases. We then upper bound the number of phases, and by extension number of rounds, needed for all nodes to wake up their neighbors. By having all nodes run \emph{Wakeup} for this many phases, we ensure that all nodes are woken up at the end of \emph{Wakeup}.
  
 We develop a series of lemmas that show that if we run \emph{Wakeup} for an unbounded number of phases, eventually all nodes will be woken up. The following fact brings out some properties of epochs.

\begin{fact}\label{fact:about-xi}
Let $t_i = 2^i \cdot 1020 c_1 \lg N$ be the size of an epoch in slot $i$. Recall that an $(N,c)$-ssf is of size $c_1 \lg N$, where $c_1$ is a constant.
\begin{enumerate}
\item $t_i$ is the number of phases sufficient to run \emph{Tree-Grower}, \emph{Tree-Cutter}, and \emph{Token-Passing-Transfer} for a given node in slot $i$.
\item $t_{i+1} = 2 t_i$.
\item Any node can only start execution of \emph{Tree-Grower} in slot $i$ at phase intervals of $t_i$ starting at phase $0$.
\end{enumerate}
\end{fact}

We now develop some terminology. Define an \textbf{activated set in phase $p$, slot $i$} to be a set of nodes such that all nodes in the set are awake in phase $p$ and have not yet completed \emph{Token-Passing-Transfer} for slot $i$, i.e. a set of nodes with status $\in [2,5]$. An \textbf{activated connected set in phase $p$, slot $i$} is an activated set of nodes in phase $p$, slot $i$ such that there exists a path between any two of its nodes in the subgraph induced by the set on the communication graph. Note that a set of nodes may be an activated set or activated connected set in a slot $i$ even if a subset of the nodes completed \emph{Token-Passing-Transfer} in a slot $<i$. We say that an activated connected set of nodes \textbf{successfully completes \emph{Token-Passing-Transfer}} if all of its neighbors are woken up. We now talk about when a node will successfully execute \emph{Token-Passing-Transfer}.

\begin{lemma}\label{lem:act-set-success}
If the number of nodes in an activated connected set $A$ is $\leq 2^i$ at the beginning of a new epoch $e_i^j$, for some $j$, then all nodes in $A$ will successfully complete execution of \emph{Token-Passing-Transfer} in $e_i^j$.
\end{lemma}

\begin{proof}
If any new nodes are awakened after the start of epoch $e_i^j$ in slot $i$, then they must wait until $e_i^{j+1}$ before they can start executing \emph{Tree-Grower}. Meanwhile, the nodes in $A$ will go on to successfully execute \emph{Token-Passing-Transfer} without interference from any node not in $A$.
\end{proof}

\begin{lemma}\label{lem:slot-find}
For any node that is woken up at phase $p$, there exists a phase $q_i \geq p$  with corresponding slot $i$ at which the size of the activated connected set to which the node belongs in $q_i$ is $\leq 2^i$.
\end{lemma}

\begin{proof}
The maximum size of any activated connected set is $n$, which is covered by slot $\lfloor \lg n \rfloor + 1$, so there will always exist a slot $i$ for any given phase where the size of the activated connected set will be $\leq 2^i$.
\end{proof}

 By Lemma~\ref{lem:act-set-success} and Lemma~\ref{lem:slot-find}, we have it that any node which wakes up will eventually successfully execute \emph{Token-Passing-Transfer} and wake up any neighboring nodes. Since the underlying communication graph is connected and initially there is at least one awake node, there exists a path to any node from an awake node and thus after enough time, all nodes in the graph will be woken up.
 Thus, it is clear that if we ran \emph{Wakeup} for an unbounded period of time, eventually all nodes would wakeup. We now bound the time taken to wake up all nodes. 

We now borrow some terminology from \cite{CGR02}.\footnote{Note that while we may incorporate some of their terminology, the technique to bound running time as seen Theorem 1 of their paper is different from ours.} A node is said to be a \textbf{frontier node} if it is awake and has neighbors that are not yet awake. A node is said to be an \textbf{inner node} if it is awake and all its neighbors are awake.

It is clear that once all nodes become inner nodes, all nodes are successfully woken up. We now bound the time it takes for nodes to become inner nodes. First, we note the following lemma.

\begin{lemma}\label{lem:unique-epoch}
Consider an activated connected set of $S$ nodes at the beginning phase $p$ of an epoch in slot $i$ such that $2^{i-1} < |S| \leq 2^i$. Let $S_f \subseteq S$, $|S_f| > 1$, be a set of frontier nodes. For each node in $S_f$, this epoch is unique and after the completion of it the node becomes an inner node.
\end{lemma}

\begin{proof}
It is clear that all nodes in $S$ will successfully complete \emph{Token-Passing-Transfer} and thus all neighbors of the frontier nodes will be woken up. We need to show, however, that the epoch satisfying the conditions of the lemma is unique for each node in $S_f$. We show this by contradiction. 

Consider a node $u \in S_f$ with two epochs $e_a^b$ and $e_x^y$ that satisfy the conditions of the lemma. Each epoch must satisfy the conditions that (i) this is an epoch in which $u$ successfully completes \emph{Token-Passing-Transfer} and (ii) $u$ is a frontier node at the start of phase of the epoch.

Consider that $a=x$, i.e. both epochs occur in the same slot. If this is the case, then property (ii) is violated. So the epochs must occur in different slots. Let us assume without loss of generality that $a < x$.

Consider that these epochs $e_a^b$ and $e_x^y$ start at different phases $p$ and $p'$ respectively. If $p < p'$, then property (ii) will be violated at the start of $p'$. It is not possible for $p > p'$ because a given node will always start execution of a \emph{Tree-Grower} the first chance it gets and there are more chances in smaller slots than larger slots. If $p=p'$, then property (i) will be violated. Since the number of nodes in $S$ remains the same regardless of which slot the epoch occurs in, there exists exactly one epoch such that $2^{i-1} < |S| \leq 2^i$.

Thus we arrive at a contradiction. It is clear that the Lemma holds true.
\end{proof}

Now, consider an epoch that satisfies the properties of Lemma~\ref{lem:unique-epoch}. If $|S_f| > |S \setminus S_f|$, call the epoch a \textbf{good epoch}. Else, call it a \textbf{bad epoch}. We upper bound the number of good epochs it would take to ensure all nodes successfully became inner nodes. We subsequently upper bound the number of bad epochs that can occur. Finally, we argue that there exists an ``uninterrupted sequence of good and bad epochs", which we formally define soon, from the start of \emph{Wakeup} until all nodes are awake. Thus, by accounting for both good epochs and bad epochs, we arrive at an upper bound on running time for the algorithm to ensure that all nodes successfully wake up.

At the end of a good epoch in slot $i$, it is clear that at least $\max \lbrace 2^{i-2}, 1 \rbrace$ frontier nodes become inner nodes. Thus after $8nt_1$ phases, assuming no bad epochs occur, all frontier nodes would become inner nodes.

At the end of a bad epoch, it is clear that at least $2^{i-1}$ nodes successfully completed \emph{Token-Passing-Transfer} in slot $i$ and will not execute \emph{Tree-Grower} again in that slot. Thus, after $4nt_1$ phases, assuming only bad epochs occur, all frontier nodes would become inner nodes.

Now, we show that there exists an unbroken sequence of good and bad epochs from the start of \emph{Wakeup} until all nodes are awake. We now formally define ``an unbroken sequence of good and bad epochs". For the rest of this argument, we say a Lemma~\ref{lem:unique-epoch} epoch when we want to refer to either a good or a bad epoch.  Consider an execution of \emph{Wakeup} that finished and note the Lemma~\ref{lem:unique-epoch} epochs of that execution, in increasing order of phase. Consider two such Lemma~\ref{lem:unique-epoch} epochs $e_a^b$ and $e_x^y$ such that $e_a^b$ starts in an earlier phase than $e_x^y$ and $e_a^b$ ends in phase $p$ and $e_x^y$ starts in phase $p'$. We say that these two epochs are in an unbroken sequence iff $p' \geq p$ and no epoch occurred in slot $x$ before $e_x^y$ in the phase period $[p,p']$. Intuitively, this means that $e_x^y$ was the next possible epoch to run in slot $x$ after $e_a^b$ finished. Since only frontier nodes have neighbors that are not awake, only in Lemma~\ref{lem:unique-epoch} epochs can those nodes wake up. Initially, there is at least a single frontier node. And since the underlying communication graph is connected, there will always exist frontier nodes until all nodes are awake. Thus, there exists an unbroken sequence of good and bad epochs from the start of \emph{Wakeup} until all nodes are awake.

We bounded the number of phases needed by good epochs to wake up all nodes $p_{good}$, assuming no bad epochs occur. We similarly bounded the number of phases needed by bad epochs to wake up all nodes $p_{bad}$, assuming no good epochs occur. However, during the execution of the algorithm, some good and some bad epochs may occur. We do not know the exact number of good epochs and bad epochs that make up this execution, but it appears that the sum $p_{good} + p_{bad}$ serves as an upper bound on the number of phases needed by the execution. However, as we shall shortly see, this is not quite the correct bound.

Before we proceed with the run time analysis, we need to take a closer look at timing issues. Consider a given node that wakes up and its unique epoch, as defined in Lemma~\ref{lem:unique-epoch}, is in slot $i$. At most it must wait $t_i$ phases before it can start executing its unique epoch, which also takes $t_i$ phases to execute. Thus, the time required to execute all good and bad epochs will be at most $2*(p_{good} + p_{bad})$.

Thus, the time taken to for all frontier nodes to successfully execute \emph{Token-Passing-Transfer}
\begin{align*}
&\leq 2*(8nt_1 + 4nt_1) \\
&= 24nt_1 \text{ phases.}
\end{align*}

Thus, if all nodes have successfully completed execution of \emph{Token-Passing-Transfer} by the end of the $24 n t_1^{\mbox{th}}$ phase, it implies that all nodes are awake after that phase. Since Algorithm \emph{Wakeup} runs for that many phases, it is guaranteed that all nodes are awake by the end of the algorithm.

As for the running time, Algorithm \emph{Wakeup} runs for $24 n t_1 = O(n \lg N)$ phases. Each phase consists of $3 (\lfloor \lg n \rfloor + 1)$ rounds. Thus, the total running time of Algorithm \emph{Wakeup} is $O(n \lg N \lg n)$ rounds.
\qed

\subsection{Tree-Grower}\label{subsect:tg}


  \emph{Tree-Grower} is given a network with an arbitrary number of awake nodes, with no knowledge of each other. It creates a forest of directed trees as follows. Initially, each node has no idea about its neighbors. The procedure takes two nodes which are close enough and silences one of the two. The silenced node becomes the child of the other node. This process is repeated until no more transmitting nodes are within range of each other. Progress occurs because there always exist two nodes which can communicate bidirectionally with each other. In the process, a forest of trees is created.

  We say that two nodes engage in \textbf{bidirectional communication} with each other when each node hears the message of the other. Define one \textbf{phase} of the \emph{Tree-Grower} procedure as the execution of the $4$ ssf schedules required by a node to engage in bidirectional communication with another node and possibly become its parent or child. A node uses the first $2$ ssf executions to determine who it can bidirectionally communicate with. The third and fourth ssf executions are used to decide which one of them becomes child and which one becomes parent (or neither). Nodes with smaller labels become parents. Amongst two nodes who are vying to become the parent of a node, the node which is heard first becomes the parent. A node is said to be \textbf{active} if it has not been silenced so far.

  The \textbf{leader} of the entire tree is the root of that tree. The procedure forces at least two nodes in each phase to engage in bidirectional communication with each other until all nodes have either become children or will become leaders. After $n$ phases of the procedure are complete, all those nodes which are not children are anointed as leaders. We restate Theorem~\ref{the:tg-props} and prove it now.

\begin{reptheorem}{the:tg-props}
	Assuming that the number of nodes given as input is an upper bound on the actual number of awake nodes in the network, the execution of Procedure \emph{Tree-Grower} on this set of nodes results in creating a forest in $O(n \lg N)$ rounds with the following properties:
	\begin{enumerate}
		\item Every node is either a leader or a child and belongs to exactly one tree.
		\item Every tree has exactly one leader and there is at most one leader per grid box of the pivotal grid.
		\item Every node knows the labels of its parent node and children.
	\end{enumerate}
\end{reptheorem}

\alglanguage{pseudocode}
\begin{algorithm}
\caption{Tree-Grower(No. of nodes $n$), run by each node $u$}
\label{prot:tree-grower}
\begin{algorithmic}[1]
	\State $type$ denotes the type of node, either $leader$ or $neutral$.
	\State $my\_leader$ denotes the parent of a node. For a root of a tree, $my\_leader$ is its own label.
	\State Tree $T$ stores the children of the node.
	\State Set $type \gets leader$ and $my\_leader \gets u$ \Comment{Initially all nodes are leaders.}
	
	\State
	
	\BlockOn {Execute the following $n$ times}
		\If {$u$ is a leader}
			\State Reset values of potential children and potential parent
			\State Execute($N,c$)-SSF: Transmit $u$'s label and store labels heard from others in $i\_hear$
			\State Execute($N,c$)-SSF: Transmit values in $i\_hear$ and store other nodes' $i\_hear$ info
			\If {A node $v$ is in $u$'s $i\_hear$ AND $u$ is in $v$'s $i\_hear$} 
				\State Add $v$ to $bidir\_comm$
			\EndIf
			\BlockOn {Execute($N,c$)-SSF: Transmit $u$'s label and store labels of nodes, $v$, heard that are in $bidir\_comm$}
				\If {$v <u$ AND $u$ has no potential parent}
					\State Store $v$ as potential parent 
				\Else 
					\State Add $v$ as a potential child
				\EndIf
			\BlockOff
			\State Set $msg$ to $u$'s label, potential parent, and potential children.
			\State Execute($N,c$)-SSF: Transmit $msg$ and store list of nodes that mention $u$ in their $msg$
			\If {$u$'s potential parent, $v$, lists $u$ as a potential child}
				\State Set $my\_leader \gets v$ and $type \gets neutral$
			\EndIf
			
			\For {Each of $u$'s potential children, $v$}
				\If {$v$ transmitted that $u$ is its potential parent}
					\State Add $v$ to $T$ as a child
				\EndIf
			\EndFor
			
		\Else \Comment{$u$ became a child and is silenced.}
			\State Do the following 4 times: Execute($N,c$)-SSF: Do nothing
		\EndIf
	\BlockOff
	
	\Return $\langle T, type, my\_leader \rangle $
\Statex
\end{algorithmic}
\end{algorithm}

\subsubsection{Proof of Theorem~\ref{the:tg-props}}
We develop a claim and two lemmas which culminate in proving that at the end of the procedure, all nodes are either leaders or children.
\begin{claim}\label{claim:bidir-result}
	When two nodes engage in bidirectional communication with each other, at least one of them will become a child.
\end{claim}

\begin{proof}
	When two nodes $u$ and $v$ engage in bidirectional communication with each other, one of them, say $u$, will have a larger label than the other. $u$ will always become a child. Either it will become $v$'s child, or else, if it engages in bidirectional communication with another node $z$, $z < v$, as well, it will become the child of that node.
\end{proof}

\begin{lemma}\label{lem:two-nodes}
	If there exist active nodes within range of each other at the start of the current phase of the procedure, then at least two nodes will engage in bidirectional communication.
\end{lemma}

\begin{proof}
	In the given phase, there will be at least two active nodes whose distance from each other is minimum among all possible distances between active nodes. By Theorem~\ref{the:ssf-replace}, executing an $(N,c)$-ssf for an appropriate $c$ ensures that at least those two active nodes can engage in bidirectional communication with each other.
\end{proof}

\begin{lemma}\label{lem:tg-inv}
	The following invariant holds for every phase $i, 1 \leq i \leq n$ of the protocol: At the beginning of every phase $i$ of the Tree-Grower procedure, either at least $i-1$ nodes have become children or all nodes which are not children will stay leaders.
\end{lemma}

\begin{proof}
	Prior to the first phase, none of the nodes have decided to become children to other nodes. So the invariant holds.
	
	Assume that the invariant holds at the beginning of the $i^{\mbox{th}}$ phase. Now, we need to prove that it holds at the beginning of the $i+1^{\mbox{th}}$ phase. Now either the first part or the second part of the condition held at the beginning of the $i^{\mbox{th}}$ phase. Let us analyze what can happen in each case.
	\begin{enumerate}
		\item At the beginning of the $i^{\mbox{th}}$ phase, at least $i-1$ nodes are children.
		
		Until the beginning of the $i^{\mbox{th}}$ phase, there were still nodes which could become children. Let us look at the two cases, namely (i) when there still exist nodes which can become children and (ii) when no more nodes can become children.\\
		(i) By Lemma~\ref{lem:two-nodes}, there must exist two nodes within transmission range of each other who have not bidirectionally communicated with each other till now, but can communicate in the phase $i$. By Claim~\ref{claim:bidir-result}, at least one of them becomes a child in phase $i$, and thus the number of nodes which are children at the beginning of phase $i+1$ is at least $i$. Hence, the invariant holds in this case.\\
		(ii) Assume that all nodes that could become children to other nodes became children by the beginning of some phase. The number of such nodes is $i-1$. Now, at the beginning of the $i^{\mbox{th}}$ phase, since no new nodes could become children, the number of children is still $i-1$. After the phase is over, there will still be $i-1$ children and all other nodes will remain leaders. Hence, the invariant holds in this case as well.
		
		\item At the beginning of the $i^{\mbox{th}}$ phase, all nodes which are not children will stay leaders.
		
		If this condition is true at the beginning of the $i^{\mbox{th}}$ phase, it will remain true at the beginning of the $i+1^{\mbox{th}}$ phase and thus the invariant will hold true at the beginning of that phase.  The only way the invariant will not hold is if two leaders engage in bidirectional communication now and one of them becomes a child. However, by Lemma~\ref{lem:two-nodes} and Claim~\ref{claim:bidir-result}, if there still exist nodes which can become children, they should've become children in the previous phases. Since no new node becomes a child, the condition remains as it was previously. Hence, the invariant still holds.
	\end{enumerate}
\end{proof}

By Lemma~\ref{lem:tg-inv}, at the end of the Tree-Grower procedure, every node has either become a child or a leader. Note that it is possible for some leaders to have no children, but every child must have a parent. As for the remaining properties:
\begin{enumerate}
	\item Since a node becomes a leader only if it has no parent at the end of the Tree-Grower procedure, the only way for one tree to have more than one leader is if a child chooses two nodes as its parents. However, this is impossible because each child chooses only one parent for itself. Also, since a child chooses only one parent, it may belong to only one tree. Hence, each node belongs to exactly one tree.
	\item If there exist two nodes within a grid box of the pivotal grid, it means that they are within range of each other. If there are two leaders in one grid box of a pivotal grid after the algorithm is over, it means that two active nodes within range of each other did not participate in bidirectional communication with each other. By Lemma~\ref{lem:two-nodes} and Lemma~\ref{lem:tg-inv}, and the fact that there are $n$ phases, there exists a phase of the procedure in which they will bidirectionally communicate with each other and one of them will become a child.
	\item Any node becomes a child of another node only when it communicates with that node and both agree to the relationship. Therefore, no node can enter into a parent-child relationship with another node unless both nodes are aware of the existence of such a relationship.
\end{enumerate}

The running time of \emph{Tree-Grower} can easily be calculated. We run a for loop $n$ times. Within this for loop, we run 4 ssf schedules, each of which takes $O(\lg N)$ rounds. Therefore, the total running time is $O(n \lg N)$.
\qed

\subsection{Tree-Cutter}\label{subsect:tc}


\emph{Tree-Cutter} takes the forest of trees created by \emph{Tree-Grower} and reduces the height of every tree to at most one, i.e. forms stars, while creating new stars if necessary.

For a given leader, each of its immediate children is called its \textbf{follower}.\footnote{At the end of \emph{Tree-Grower}, each tree may have a single leader and multiple children as well as children of children and so on. Only a subset of those children will become followers of that leader.} A node can be a follower of only one leader. A node which is not a leader is said to \textbf{align} itself to a leader when it hears from a leader and decides to make that leader its parent. It then becomes that node's follower. During the procedure, a node can have one of the following three statuses: leader, follower, or neutral. A \textbf{neutral node} is a node which has not aligned itself to a leader (or anointed as a leader yet). All nodes are initially neutral nodes with the exception of the roots (leaders) of trees formed by the end of \emph{Tree-Grower}. Those nodes stay leaders at the beginning of \emph{Tree-Cutter}.

  Define one \textbf{phase} of \emph{Tree-Cutter} as $5$ executions of an ssf schedule. For a group of nodes, passing a \textbf{token} amongst themselves, the token is just a message sent or received by the nodes. In a given phase, a \textbf{potential leader} is a neutral node which has a token.

  The \emph{Tree-Cutter} procedure works as follows: 
  \begin{enumerate}
  \item The first $3$ executions are uniquely used by potential leaders to elect at most one leader among themselves per grid box of the pivotal grid, using a sub-procedure \emph{Potential-Leader-Election}. 
  \item The next ssf execution is used to announce whether the node is a leader or someone's follower. 
  \item The final ssf execution is used to pass the token, if any, to another node as follows. If the node is a neutral node or has no children, then it passes the token back to the node it received the token from. If the node has any children, then it passes the token to one of them in Depth First Search (DFS) manner. Note that, if the node became a leader during the execution of \emph{Potential-Leader-Election}, then it generates and keeps a new token with itself. 
  \end{enumerate}
  
  It is important to note that if a node changes its alignment to a new leader, then once it gets the token that ``belongs" to its old leader and announces its change, it returns that token back to the node that sent it. This is a key part of the algorithm that allows all nodes to eventually get the token.

  The \emph{Potential-Leader-Election} sub-procedure takes place amongst neutral nodes who receive a token. They initially are potential leaders and transmit that during an ssf execution. By the property of a strongly selective family, and since at most a constant number of tokens are present in each grid box, there exists a round for each token holder in which only it will transmit, successfully, to all others within its range. The nodes then execute a second ssf schedule and transmit the information of which nodes they heard in each round of the ssf schedule.\footnote{The idea to transmit a bit string of the rounds that each potential leader was heard in, for each such leader, is inspired from \cite{JRS17}. In that paper, they mention a similar technique to optimize the size of local memory of nodes to execute the procedure ProximityGraphConstruction.} Each token holder uses this information (and ssf schedule) to figure out which round of the ssf schedule was the round in which only it successfully transmitted. Then, one final ssf schedule is executed where the node transmits at most once during the entire execution. In this execution, if the number of the current round is equal to the round number in which the given node is the sole potential leader within its range to transmit, then it transmits and becomes a leader. If before its round comes up, it hears another potential leader transmit, then it becomes a follower of that node. Once it becomes either a leader or a follower, it remains inactive for the remaining rounds. After all rounds are over, all nodes are either leaders or followers. Furthermore, if a node has become a leader, then every potential leader, within its range, who was not already a follower must become its follower. Since all nodes within a grid box are within range of each other, there can exist at most one leader per grid box.

  Now, we restate Theorem~\ref{the1} and prove it.

\begin{reptheorem}{the1}
  Assume that the nodes satisfy the properties described in Theorem~\ref{the:tg-props}. The execution of Procedure \emph{Tree-Cutter} takes $O(n \lg N)$ rounds and results in the creation of stars that satisfy the following properties:
\begin{enumerate}
\item Every node is either a leader or a follower and belongs to exactly one tree.
\item Each star has exactly one leader, which is the center of the star.
\item There exists at most one leader node per grid box of the pivotal grid.
\end{enumerate}
\end{reptheorem}

\alglanguage{pseudocode}
\begin{algorithm}
\begin{algorithmic}[1]
\caption{Potential-Leader-Election(Type $type$), run by each node $u$}
	\State $type$ denotes the type of node, either $leader$, $follower$, or $neutral$.
	\State $my\_leader$ denotes leader of node. For a leader, $my\_leader$ is its own label.
	
	\State
	
	\If {$u$ is a neutral node and has a token}
		\State Execute($N,c$)-SSF: Transmit that $u$ is a potential leader and store list of each node, $v$, who is a potential leader and what round it transmitted in
		\State For each potential leader, $v$, that was heard, construct the bit string $rounds\_heard(v)$ of size $c_1 \lg N$ from left to right. Each round that $v$ was heard has a $1$ in the corresponding position and the remaining bits are $0$
		\State Execute($N,c$)-SSF: Transmit all pairs of $\langle v,rounds\_heard\rangle $ heard
		\State \small Find which round, $R$, $u$ alone was heard by every potential leader within range \normalsize
		\BlockOn{Execute ($N,c$)-SSF: In each round of the execution, do the following:}
			\If{$u$ is a neutral node}
				\If{current round $= R$}	\Comment{$u$ to transmit that $u$ is a leader.}
					\State Transmit that $u$ is the leader
					\State $type \gets leader$, $my\_leader \gets u$
				\Else
					\State \small If $u$ hears $v$ say it is leader, then $type \gets follower$, $my\_leader \gets v$ \normalsize
				\EndIf
			
			\Else 	\Comment{$u$ became either a leader or a follower.}
				\State Do nothing
			\EndIf
		\BlockOff

	\Else \Comment{$u$ doesn't participate in leader election.}
		\State Do the following $3$ times: Execute($N,c$)-SSF: Do nothing
	\EndIf
	
	\Return $\langle type, my\_leader\rangle $
\end{algorithmic}
\end{algorithm}

\alglanguage{pseudocode}
\begin{algorithm}
\caption{Tree-Cutter(Tree $T$, Type $type$, Node $my\_old\_parent$, No. of nodes $n$), run by each node $u$}
\label{prot:tree-cutter}
\begin{algorithmic}[1]	
	\State $type$ denotes the type of node, either $leader$, $follower$, or $neutral$.
	\State \small $my\_old\_parent$ stores parent of node initially. For leader, this value is own label. \normalsize
	\State $my\_leader$ stores leader of a node. For a leader, $my\_leader$ is its own label. 
	\State Tree $T$ stores the children of the node. 
	
	\State
	
	\If{$u$ is a leader}
		\State Set $my\_leader \gets u$ and create a token
	\EndIf
	
	\BlockOn {Execute the following $8n$ times}
		\BlockOn{\textbf{Case 1:} $u$ is a leader or follower}
			\State Potential-Leader-Election($type$)
			\If {$u$ has a token}
				\State Execute($N,c$)-SSF: Transmit $u$'s label, $type$, and $my\_leader$. If $u$ is a leader, update $T$ by adding any nodes that have made $u$ their leader and removing those that made someone else their leader
				\If {$u$ is a leader}
					\State $target \gets$ one of the children of $T$, chosen in DFS manner.
				\ElsIf {$u$ had received the token from $u$'s leader}
					\State $target \gets$ one of the as yet untraversed children of $u$'s tree $T$
				\Else \Comment{$u$ had received the token from $my\_old\_parent$.}
					\State $target \gets my\_old\_parent$
				\EndIf
				\State Execute($N,c$)-SSF: Pass token to $target$
			\Else
				\State Execute($N,c$)-SSF: If $u$ is a leader, add any nodes that made $u$ their leader as children in $T$
				\State Execute($N,c$)-SSF: Listen to see if $u$ receives a token
			\EndIf
		\BlockOff
	
		\BlockOn{\textbf{Case 2:} $u$ is a neutral node}
			\If {$u$ has a token}
				\State $\langle type, my\_leader\rangle  \gets$ Potential-Leader-Election($type$)
			
				\If {$u$ is a leader}
					\State Execute($N,c$)-SSF: Transmit that $u$ is a leader
					\State Execute($N,c$)-SSF: Pass token back to whomever sent it to $u$
					\State Generate and keep a new token for $u$
				\Else \Comment{Potential leaders who became followers instead.}
					\State Execute($N,c$)-SSF: Transmit that $u$ is a follower of $u$'s leader
					\State Execute($N,c$)-SSF: Pass token back to whomever sent it to $u$
				\EndIf
			\Else
				\State Potential-Leader-Election($type$)
				\BlockOn {Execute($N,c$)-SSF:}
					\If {$u$ hears node $v$ say it is leader AND $type(u) = neutral$}
						\State $my\_leader \gets v$, $type \gets follower$
					\EndIf
				\BlockOff
				\State Execute($N,c$)-SSF: Listen to see if $u$ receives a token
			\EndIf
		\BlockOff
	\BlockOff
	
	\State \textbf{if} $u$ is a follower \textbf{then} Reset $T$
	
	\Return $\langle T, type, my\_leader\rangle $
\Statex
\end{algorithmic}
\end{algorithm}

\subsubsection{Proof of Theorem~\ref{the1}}

\noindent We first prove that at end of the procedure, there exists at most one leader per grid box of the pivotal grid. We then prove that at the end of the procedure, all nodes will be either leaders or followers. Subsequently, we prove that all trees formed and created are stars at the end of the procedure. Finally, we bound the running time of \emph{Tree-Cutter}.

  We now prove that there exists at most one leader per grid box of the pivotal grid. We show this using a strong induction proof on two invariants, where we show that if both invariants hold at the beginning of a given phase and every prior phase of \emph{Tree-Cutter}, then they will hold at the beginning of the next phase.\\
  $Inv_1$: At the beginning of any phase of \emph{Tree-Cutter}, at most one leader can exist in a grid box of the pivotal grid.\\
  $Inv_2:$ At the beginning of any phase of \emph{Tree-Cutter}, at most a constant number of tokens are present in any grid box of the pivotal grid.

We now assume $Inv_1$ and $Inv_2$ hold true at the beginning of some phase. We prove that $Inv_1$ holds true in the next phase. 
This is captured by Lemma~\ref{lem4}, that for any grid box of the pivotal grid, at most one node can become a leader. Every leader is either already present at the beginning of \emph{Tree-Cutter} or created using \emph{Potential-Leader-Election}. We first show in Lemma~\ref{lem3} that assuming a constant number of token holders (and by extension potential leaders) exist in any grid box of the pivotal grid at the beginning of a phase of \emph{Tree-Cutter}, \emph{Potential-Leader-Election} elects at most one leader per grid box of the pivotal grid containing potential leaders.

\begin{lemma}
\label{lem3}
  If at most a constant number of potential leaders exist in any grid box of the pivotal grid at the beginning of a phase of \emph{Tree-Cutter}, then at the end of the \emph{Potential-Leader-Election} procedure in that phase, all those potential leaders become either leaders or followers. Furthermore, at most one leader will be created in those grid boxes that previously contained potential leaders.  
\end{lemma}

\begin{proof}
For each potential leader, in the first ssf schedule, call the round in which it was heard by all other potential leaders within range its \textbf{main round}. Define \textbf{participating nodes} of a given phase as all potential leaders and those potential leaders which have become leaders or followers. For a given phase of \emph{Tree-Cutter}, we show that the following invariant holds for every participating node at the beginning of each round of the execution of the third ssf schedule in \emph{Potential-Leader-Election}: \textbf{At most one participating node will become a leader per grid box of the pivotal grid. All participating nodes within range of a created leader are followers, though not necessarily of that leader.}

Prior to the first round, since no potential leader has become a leader, the invariant holds true.

Assume that the invariant holds true up to the current round. If no node transmits in the current round, then the invariant continues to hold. 

For a given grid box of the pivotal grid, if in the current round, a node within that grid box is the sole transmitter within that grid box, then all other potential leaders in that box become followers of that node since there are a constant number of potential leaders in any grid box and by Theorem~\ref{the:ssf-replace} they hear the transmitter. If there were followers of other nodes within that grid box, they remain of the same type. So the invariant holds at the beginning of the next round.

Thus, after all rounds of the ssf schedule are executed, at most one of the potential leaders becomes a leader per grid box of the pivotal grid and all participating nodes within its range are followers, though not necessarily of it. Thus, all potential leaders in any given grid box have either become a leader or followers at the end of the \emph{Potential-Leader-Election} procedure.
\end{proof}

Now, we are ready to prove that $Inv_1$ holds true at the beginning of the next phase.

\begin{lemma}\label{lem4}
If for all phases $\leq i$ of \emph{Tree-Cutter}, $Inv_1$ and $Inv_2$ hold true, then for each grid box of the pivotal grid, at most one node can become a leader in phase $i$. Furthermore, once a leader is elected in a grid box, all other nodes within its range and within its grid box must become its followers or followers of another leader.
\end{lemma}

\begin{proof}
  There are only two ways in which leaders may arise in a grid box:
\begin{enumerate}
\item Originally a leader existed in that grid box.\\
At the start of \emph{Tree-Cutter}, every grid box of the pivotal grid has at most one leader by Theorem~\ref{the:tg-props}. This leader will transmit that it is a leader and all nodes within range (other nodes in the grid box) will become followers of it if they were not already someone else's followers. This is because other nodes will hear the transmission of the leader due to Theorem~\ref{the:ssf-replace}.

\item A leader was elected by the \emph{Potential-Leader-Election} procedure.\\
If a grid box of the pivotal grid has a leader, then it will announce its status when it was created and all surrounding nodes in the grid box will not be able to become potential leaders. If a grid box does not have a leader yet and tokens make their way to neutral nodes within that grid box, then Potential-Leader-Election will be run on those nodes. By Lemma~\ref{lem3}, we are guaranteed that at most one leader will be elected per grid box and the remaining potential leaders within range will become followers. Then, during the rest of \emph{Tree-Cutter} in that phase, the newly elected leader will transmit its status; since there are a constant number of token holders, and by extension transmitters in any grid box of the pivotal grid, Theorem~\ref{the:ssf-replace} guarantees that other nodes within range will successfully hear its message and become followers of it if they were not already followers of another leader.
\end{enumerate}
\end{proof}

  We now assume $Inv_1$ and $Inv_2$ hold true at the beginning of some phase. We prove that $Inv_2$ holds true in the next phase. This is captured by Lemma~\ref{lem7}, that at the beginning of each phase of \emph{Tree-Cutter}, the number of tokens present in a grid box of the pivotal grid is at most some constant number. We do this by showing that the number of tokens that can travel to any grid box is a constant and the number of tokens that can be created in a given grid box is a constant. Thus, at any time, there is only a constant number of tokens in any grid box of the pivotal grid.

  To show that only a constant number of tokens may travel to any grid box of the pivotal grid, we argue that the farthest distance a token can travel is at most $2r$ away from its leader. This is assuming that $Inv_2$ holds in the phases that the token moves away from its leader.

\begin{lemma}\label{lem5}
  Consider a leader $x$ and the token it generated $T_x$. Assume that at the beginning of some phase $j$, $T_u$ is present at node $u$. Now, $T_u$ can travel a distance at most $2r$ away from $u$, assuming that $Inv_2$ holds in phases $j$ and $j+1$.
\end{lemma}

\begin{proof}
  We first show that a token $T_x$ can only move at most two hops away from its leader $x$. We then bound the distance the token can travel from its leader.  
  
  We now argue about the maximum number of hops away from $x$ that $T_x$ can move. First, we constrain the types of nodes a $x$ can pass $T_x$ to. We then constrain the kinds of nodes a $T_x$ can be sent to in the second hop. We show that $T_x$ will never make a third hop away from the $x$ but back towards it. That is to say, if $T_x$ was passed from node $x$ to node $y$ in the first hop and from node $y$ to node $z$ in the second hop, then $T_x$ will be passed back to node $y$ in the third hop.

  The types of nodes that can exist in the network are leaders, neutral nodes, and followers. It is clear from the algorithm construction that $x$ will not pass $T_x$ to another leader or a neutral node. Thus, $x$ will only pass its token to a follower $y$. 
  
   A follower $y$ is one of two kinds. Either (i) $y$ is a follower of another leader $w$ and is $x$'s child in its tree $T$ but has not yet transmitted that it is $w$'s follower or (ii) $y$ is a follower of $x$ that is either $x$'s immediate child in $x$'s tree formed at the end of \emph{Tree-Grower} or another node which has already transmitted that it is a follower of $x$. The first kind of follower will announce its realignment to $w$ and send the token back to $x$. The second kind of follower will pass its token to a child $z$ within its sub-tree, if it has a sub-tree. 
   
   The child $z$ will either be of type follower or neutral. It will not be of type leader, since newly created leaders announce their new status immediately and cease to belong to their old sub-tree. If $z$ is of type follower, then $z$ will transmit its realignment and return the token to $y$. If $z$ is of type neutral, then after the potential leader election is done, $z$ will either be a follower of another leader or a leader itself. In either case, $z$ will return the token back to $y$. 
   
   Thus, any token can move at most two hops away from its leader. Since the maximum distance of each hop is at most $r$, any token can travel a distance at most $2r$ away from its leader. Note that in both phases the token is passed, $Inv_2$ holds. $Inv_2$ guarantees that the number of nodes that transmit within any box of the pivotal grid is bounded by a constant. Thus, by Theorem~\ref{the:ssf-replace}, the passing of the token (message transfer) is successful. 
\end{proof}

  We are now ready to prove that $Inv_2$ holds at the beginning of the next phase. 

\begin{lemma}\label{lem7}
 If for all phases $\leq i$ of \emph{Tree-Cutter}, $Inv_1$ and $Inv_2$ hold true, then at the beginning of phase $i+1$ the number of tokens present in a grid box of the pivotal grid is at most some constant number.
\end{lemma}

\begin{proof}
  At the beginning of the $i^{\mbox{th}}$ phase, we begin with a constant number of tokens per grid box. Now, we need to show that at the beginning of the $i+1^{\mbox{th}}$ phase, we also begin with a constant number of tokens. 

  By Lemma~\ref{lem5}, the maximum distance a token can travel is at most $2r$ away from its leader. For a given grid box $b$, there are $44$ surrounding grid boxes from which a token may make its way into $b$ by traveling a distance of at most $2r$. Therefore, only a constant number of tokens can enter any grid box by the beginning of phase $i$. By Lemma~\ref{lem4}, a grid box can have at most only one new leader (and by extension token) created within its boundaries. Therefore, the total number of tokens that can be present within a grid box at the beginning of the next phase is some constant number.
\end{proof}

Prior to the start of \emph{Tree-Cutter}, by Theorem~\ref{the:tg-props}, it is clear that both $Inv_1$ and $Inv_2$ are true. Lemma~\ref{lem4} guarantees that so long as we start any phase with a constant number of token holders per grid box, the number of leaders elected per grid box is at most one. By Lemma~\ref{lem7}, the number of the token holders in any grid box at the beginning of any phase of the \emph{Tree-Cutter} procedure is at most a constant number. Thus, after the execution of the \emph{Tree-Cutter} procedure, there exists at most one leader per grid box of the pivotal grid. 

  We now prove that all nodes eventually become either leaders or followers.
  
\begin{lemma} By the end of the execution of \emph{Tree-Cutter}, all nodes become either leaders or followers. 
\end{lemma}

\begin{proof}
At the start of \emph{Tree-Cutter}, each node is either a leader or a neutral node. Each node that is a leader stays a leader and each node that is a neutral node eventually becomes either a leader or a follower.

We bound the number of times any node receives a token until it eventually becomes a leader or a follower. We then bound the number of phases required to ensure that all nodes are able to receive a token the required number of times.

Each non-leader $y$ receives a token at most $4$ times, excluding the times where $y$ gets a token back from its children:
\begin{enumerate}
\item When $y$ is a neutral node which got a token from its parent in the original forest. 
\item When $y$ is a follower and receives a token from its parent in the original forest.
\item When $y$ is a follower that has changed parents (either $y$ was a descendant of its leader $x$ from the original forest and became $x$'s direct child or else $y$ changed leaders to $z$) and receives a token from its leader.
\item When $y$ becomes a leader, i.e. when $y$ was a neutral node that became a leader and generated a token.
\end{enumerate} 

Each leader $x$ receives a token exactly once, excluding the times where $x$ gets a token back from its children or followers. The one time is when $x$ initially generates a token.

Notice that each time a node receives a token, as described above, it makes progress towards becoming either a follower or a leader. Define a \textit{progress event} as either the occurrence when a node $y$ receives a token as described above or the occurrence when $y$ sends a token to $y$'s parent.

Each node $y$ receives a token at most $4$ times, excluding the times when $y$ gets a token back from its children or followers. Notice that we can account for token passes from a child/follower to its parent as a message sent by the child. Thus, each node $y$ receives the token from a parent at most $4$ times and sends a token back to $y$'s parent at most $4$ times, i.e. each node has at most $8$ progress events associated with it.

It is clear to see from the construction of \emph{Tree-Cutter} that in each phase, at least one progress event occurs. There are $n$ nodes, each requiring at most $8$ progress events. Thus, by allowing \emph{Tree-Cutter} to run for $8n$ phases, we allow sufficient time for all nodes to eventually become leaders or followers.
\end{proof}
  
  Since nodes can be followers of only one leader, each node belongs to exactly one tree by the end of \emph{Tree-Cutter}. It is clear from the construction, that every tree is cut down to a star with exactly one root node.

  Regarding the running time, there are $O(n)$ phases. Each phase of the procedure consists of 5 executions of an ssf schedule. Each  execution of an ssf schedule takes $O(\lg N)$ rounds to complete. Thus, the entire procedure takes $O(n \lg N)$ rounds to complete. This concludes the proof of Theorem~\ref{the1}.
\qed
\\

\noindent We prove one more lemma about the \emph{Tree-Cutter} procedure which will be used in a subsequent section on the backbone algorithm.

\begin{lemma}\label{lem:3-hop}
  At the end of \emph{Tree-Cutter}, assuming that there exist at least two leaders, there exists at least one other leader within 3 hops of any leader.
\end{lemma}

\begin{proof}
  Consider a leader $x$ at the end of \emph{Tree-Cutter} procedure. All children of that leader would be made followers of either that leader or another leader. If they were made children of another leader, that means that in 2 hops from $x$ there exists a leader. Let us assume that this is not the case. Now consider the followers of $x$. It must occur that there exists one node within range of one of the followers of $x$ or else there are no more nodes. If this is not the case, it means that there exist nodes which have no means of communicating with $x$ and this implies that the communication graph is disconnected. But according to our model, that is not the case. Therefore, let us consider that there exists at least one node within range of one of the followers assuming there exist nodes which are not followers of $x$. The node can only be a leader or a follower by Theorem~\ref{the1}. Then within 3 hops of $x$, there exists a leader. Therefore, for every leader, there exists at least one other leader within 3 hops of that leader.
\end{proof}

\subsection{Token-Passing-Transfer}\label{subsect:tpt}


The goal of \emph{Token-Passing-Transfer} is to ensure that every participating node's message is successfully transmitted. When used for wakeup, the message that needs to be transmitted is a special \textit{wakeup} message (and some additional info) which allows neighboring nodes to wakeup. Here, all participating nodes are transmitting the same message and so all that is required is that all participating node successfully transmit this wakeup message once. We restate Theorem~\ref{the:tpt-works} and prove it now.

\begin{reptheorem}{the:tpt-works}
	Assuming that the participating nodes are in a configuration characterized by Theorem~\ref{the1}, the \emph{Token-Passing-Transfer} procedure takes $O(n  \lg N)$ rounds to complete and ensures  that each participating node successfully transmits the message.
\end{reptheorem}

\alglanguage{pseudocode}

\begin{algorithm}
\caption{Token-Passing-Transfer(Tree $T$, Type $type$, Node $my\_leader$, No. of nodes $n$, Message $msg$), run by each node $u$}
\label{prot:token-passing-transfer}
\begin{algorithmic}[1]
	\State $type$ denotes the type of node, either $leader$ or $follower$.
	\State $my\_leader$ stores leader of a node. For a leader, $my\_leader$ is its own label.
	\State 	Tree $T$ stores the children of a leader. For followers, $T$ is empty.
	\State
	\State \textbf{if} $u$ is a leader \textbf{then} Create token

	\BlockOn{Execute the following $488n$ times}
		\If {$u$ has a token}
			\State Execute($N,c$)-SSF: Transmit $msg$
			\State Execute($N,c$)-SSF: If $u$ is a leader and has children, pass token in DFS manner to $u$'s children. Else pass token to $my\_leader$
		\Else
			\State Execute($N,c$)-SSF: Do nothing
			\State Execute($N,c$)-SSF: Listen to see if $u$ receives a token
		\EndIf
	\BlockOff
\Statex
\end{algorithmic}
\end{algorithm}

\subsubsection{Proof of Theorem~\ref{the:tpt-works}}
We now prove that if a wakeup message needs to be transmitted, all participating nodes will transmit the message by the end of procedure. Since tokens can be passed to nodes at most one hop away from a leader, the maximum number of nodes with tokens that can transmit from a particular grid box is bounded by a constant. Furthermore, since there are only a constant number of grid boxes within range of a grid box containing a node, the number of  nodes within range of a given node which can transmit at any given time is bounded by a constant. By Theorem~\ref{the:ssf-replace}, by choosing a sufficiently large constant $c$, using an ($N, c$)-ssf, we can guarantee that a participating node will successfully transmit its wakeup message.

As to the running time, every node executes two ($N,c$)-ssf schedules of length $O(\lg N)$ rounds $488n$ times, for a total running time of $O(n \lg N)$ rounds.
\qed


\section{Backbone Subnetwork}\label{sect:backbone}

There are three algorithms required to turn a network into a functioning backbone. Algorithm \emph{Backbone-Creation} creates the backbone. Algorithm \emph{Backbone-Message-Exchange} is used to communicate between nodes within the backbone. Algorithm \emph{Backbone-Message-Transmit} is used to communicate between nodes within the backbone and nodes outside of it. Algorithm \emph{Backbone-Message-Exchange} and Algorithm \emph{Backbone-Message-Transmit} are constantly run in a time division multiplexing fashion.

\subsection{Backbone-Creation}
Algorithm \emph{Backbone-Creation} uses the stars produced from the \emph{Tree-Cutter} algorithm and turns them into a backbone subnetwork.

\begin{theorem}\label{the:back-create-time}
Assuming all nodes are initially awake, the Algorithm \emph{Backbone-Creation} creates a backbone subnetwork in running time $O(n \lg N)$ rounds.
\end{theorem}

Define a \textbf{connector node} $v$ for a given leader $u$ as a follower of $u$ which, in the backbone, is a node in the path between $v$ and another leader at most 3 hops away from it. Algorithm \emph{Backbone-Creation} sets up all leaders as members of the backbone and then for each leader, chooses a constant size subset of its followers to be connectors.

Once \emph{Tree-Grower} and \emph{Tree-Cutter} are run, Algorithm \emph{Backbone-Creation} proceeds in 4 stages. Let us assume each stage consists of a leader passing the token to all its followers and allowing them as well as itself to transmit. The first stage consists of all followers announcing that they are followers of a given leader. The second stage consists of followers announcing their leaders and $\langle leader, follower\rangle$ pairs they hear in the previous stage. Leaders designate some their followers to act as connectors to other leaders. The third stage consists of leaders and their followers transmitting who is a connector and who is not. The final stage consists of all nodes belonging to the backbone transmitting that they belong. Thus, nodes develop a local view of the backbone. These four stages correspond to the sub-procedures \emph{Backbone-Follower-Announce}, \emph{Backbone-Follower-Extended-Announce}, \emph{Backbone-Connector-Announce}, and \emph{Backbone-Belong-Announce} respectively. When calling these sub-procedures, we omit any parameters as it is assumed that they are called within the context of \emph{Backbone-Creation} and thus have access to all relevant data.


\alglanguage{pseudocode}
\begin{algorithm}
\caption{Backbone-Follower-Announce, run by each node $u$}
\label{prot:backbone-follower-announce}
\begin{algorithmic}[1]
		\If{$u$ is a leader}
			\For{ Every child in $T$}
				\State Execute($N,c$)-SSF: Store a list of all labels of other leaders' followers and labels of their leaders that $u$ hears. Add these followers to $u$'s list of connectors
				\State Execute($N,c$)-SSF: Pass token to the child
				\State Do the following 2 times: Execute($N,c$)-SSF: Do nothing
			\EndFor
		\Else
		
			\BlockOn {Execute the following $num\_phases$ times}
				\If {$u$ has a token}
					\State Execute($N,c$)-SSF: Transmit that $u$ is a follower of $u$'s leader and store a list of all labels of other leaders' followers and labels of their leaders that $u$ hears
					\State Execute($N,c$)-SSF: Pass token back to $u$'s leader
				\Else
					\State Execute($N,c$)-SSF: Store a list of all labels of other leaders' followers and labels of their leaders that $u$ hears
					\State Execute($N,c$)-SSF: Listen to see if $u$ receives a token
				\EndIf
			\BlockOff
		\EndIf
\end{algorithmic}
\end{algorithm}

\alglanguage{pseudocode}
\begin{algorithm}
\caption{Backbone-Follower-Extended-Announce, run by each node $u$}
\label{prot:backbone-follower-extended-announce}
\begin{algorithmic}[1]
		\If{$u$ is a leader}
			\For {Every child in $T$}
				\State Execute($N,c$)-SSF: Do nothing
				\State Execute($N,c$)-SSF: Pass token to the child
				\State Execute($N,c$)-SSF: Store list of all triplets $\langle x$'s follower, other leader's follower, other leader$\rangle $ that $u$ hears from $x$'s children, $\forall x$
				\State Execute($N,c$)-SSF: Do nothing
			\EndFor
		\Else
		
			\BlockOn {Execute the following $num\_phases$ times}
				\If {$u$ has a token}
					\State Execute($N,c$)-SSF: Transmit the list of all other leader's followers and their leaders that $u$ heard in Stage 1
					\State Execute($N,c$)-SSF: Pass token back to $u$'s leader
				\Else
					\State Execute($N,c$)-SSF: Do nothing
					\State Execute($N,c$)-SSF: Listen to see if $u$ receives a token
				\EndIf
			\BlockOff
		\EndIf
\end{algorithmic}
\end{algorithm}

\alglanguage{pseudocode}
\begin{algorithm}
\caption{Backbone-Connector-Announce, run by each node $u$}
\label{prot:backbone-connector-announce}
\begin{algorithmic}[1]
		\If{$u$ is a leader}
			\State Calculate the shortest routes to every other leader $u$ has heard about. Add intermediary nodes to list of connectors
			\For {Every child in $T$}
				\State Execute($N,c$)-SSF: Transmit the list of connectors
				\State Execute($N,c$)-SSF: Pass token to the child
				\State Execute($N,c$)-SSF: If $u$ hears of a node within range of it being a connector, add it to $u$'s list of connectors
				\State Execute($N,c$)-SSF: Do nothing
			\EndFor
		\Else
		
			\BlockOn {Execute the following $num\_phases$ times}
				\If {$u$ has a token}
					\State Execute($N,c$)-SSF: Transmit the list of $u$'s leader's connectors and listen to see if $u$ is designated a connector by another leader
					\State Execute($N,c$)-SSF: Pass token back to $u$'s leader
				\Else
					\State Execute($N,c$)-SSF: Store the list of connectors $u$'s leader transmits and listen to see if another leader designated $u$ as a connector
					\State Execute($N,c$)-SSF: Listen to see if $u$ receives a token
				\EndIf
			\BlockOff
		\EndIf
\end{algorithmic}
\end{algorithm}

\alglanguage{pseudocode}
\begin{algorithm}
\caption{Backbone-Belong-Announce, run by each node $u$}
\label{prot:backbone-belong-announce}
\begin{algorithmic}[1]
		\If{$u$ is a leader}
			\State Set $u$ as root of $H$
			\For { Every child in $T$}
				\State Execute($N,c$)-SSF: Transmit that $u$ belongs to the backbone
				\State Execute($N,c$)-SSF: Pass token to the child
				\State Execute($N,c$)-SSF: Add any nodes $u$ hears from as children of $H$ and to $u$'s list of connectors if they are not already there
				\State Execute($N,c$)-SSF: Do nothing
			\EndFor
		\Else
		
			\BlockOn {Execute the following $num\_phases$ times}
				\If{$u$ is a connector}
					\If{$u$ has a token}
					\State Execute($N,c$)-SSF: Transmit that $u$ is a connector. Add any nodes $u$ hears from as children of $H$
					\State Execute($N,c$)-SSF: Pass token back to $u$'s leader
					\Else
					\State Set $u$ as root of $H$	
					\State Execute($N,c$)-SSF: Add any nodes $u$ hears from as children of $H$
					\State Execute($N,c$)-SSF: Listen to see if $u$ receives a token
					\EndIf

				\Else
					\State Execute($N,c$)-SSF: Do nothing
					\State Execute($N,c$)-SSF: If $u$ has a token, pass it back to $u$'s leader. If $u$ does not have a token, listen to see if $u$ receives one
				\EndIf
			\BlockOff
		\EndIf
\end{algorithmic}
\end{algorithm}

\alglanguage{pseudocode}
\begin{algorithm}
\caption{Backbone-Creation(No. of nodes $n$), run by each node $u$}
\label{prot:backbone-creation}
\begin{algorithmic}[1]
	\State $type$ denotes the type of node, either $leader$ or $follower$.
	\State $my\_leader$ denotes the leader of a node. For a leader, $my\_leader$ is its own label.
	\State Tree $T$ stores the children of a leader. For followers, $T$ is empty.
	\State Tree $H \gets \bot$ \Comment{$H$ is the local view of the backbone network from $u$'s perspective. If $u$ is part of the backbone, then $H$ is rooted at $u$ and its children are all nodes in the backbone network within range of $u$. If not then $H = \bot$.}

	\State
	
	\State $\langle T, type, parent\rangle  \gets$ Tree-Grower($n$)
	
	\State $\langle T, type, my\_leader\rangle  \gets$ Tree-Cutter($T$, $type$, $parent$, $n$)
	
	\State Execute($N,c$)-SSF: If $u$ is a leader, transmit the number of $u$'s children. If $u$ is a child, listen for this value from $u$'s leader.
	\State Set $num\_phases \gets 2 * (\mbox{the no. of children of $u$'s leader})$
	
	\State Backbone-Follower-Announce
	\State Backbone-Follower-Extended-Announce
	\State Backbone-Connector-Announce		
	\State Backbone-Belong-Announce

\Statex
\end{algorithmic}
\end{algorithm}

\subsubsection{Proof of Theorem~\ref{the:back-create-time}}
The proof of the theorem proceeds as follows. First we prove message transmission during the algorithm will be successful, i.e. the conditions to use SSF based dilution are satisfied in the algorithm. We then go on to prove that the algorithm creates a network that satisfies the properties of a backbone. Finally, we bound the running time of the algorithm.

The algorithm initially calls \emph{Tree-Grower} and \emph{Tree-Cutter}. As proved earlier, those procedures work as intended. We show that subsequently in the algorithm, when a node transmits a message, nodes within range of it can hear that message. 
\begin{lemma}\label{lem:ssf-works}
For Algorithm \emph{Backbone-Creation} post the call to \emph{Tree-Cutter}, using an ($N,c$)-ssf schedule for a sufficiently large constant $c$, where only nodes with tokens can transmit, allows all nodes within range of a token holder to hear its message.
\end{lemma}

\begin{proof}
We show this by showing that the number of tokens that can be active in a grid box of the pivotal grid in any round of Algorithm \emph{Backbone-Creation} post the call to \emph{Tree-Cutter} is a constant. In Algorithm \emph{Backbone-Creation}, post the call to \emph{Tree-Cutter}, every leader has one token associated with it and can only send this token to nodes one hop away from it. Since there exists at most one leader per grid box of the pivotal grid at the end of \emph{Tree-Cutter}, the number of tokens coming into a node can be at most 20 plus one token already belonging to a leader of the box. Thus, at most 21 tokens, a constant number, can be active in any grid box at any given time. By Theorem~\ref{the:ssf-replace}, we know that using an ($N,c$)-ssf schedule for a sufficiently large $c$ allows all nodes within range of a token holder to hear its message.
\end{proof}

We now prove that \emph{Backbone-Creation} indeed creates a backbone. According to Jurdzi\'nski and Kowalski~\cite{JK12}, the properties that need to be satisfied for a network to be considered a backbone are:
\begin{enumerate}
\item The nodes of the backbone, $H$, form a connected dominating set which induces a subgraph of the communication graph, $G$, and have a constant degree relative to other nodes within the backbone.
\item The number of nodes in $H$ is $O(s\_c\_d)$, where $s\_c\_d$ is the size of the smallest connected dominating set of the communication graph.
\item Each node $u$ of $G \backslash H$ is associated with exactly one node of $H$ which is also a neighbor of $u$, which acts as its entry point into the backbone.
\item The asymptotic diameter of $H$ is same as that of $G$.
\end{enumerate}

We show that each of these properties are satisfied.
\begin{enumerate}
\item We first show that $H$ is a connected dominating set in (a). In (b) we prove that the internal degree of nodes in $H$ is constant. In (c), we argue that $H$ induces a subgraph of $G$.
	\begin{enumerate}
	\item By Lemma~\ref{lem:3-hop}, if there exists more than one leader, then for every leader, there exists at least one other leader within 3 hops of it.  Since $G$ is connected, it must be the case that there exists a path between any two leaders at most 3 hops apart. Therefore, if we show that the algorithm connects every leader to all leaders within 3 hops of it, then $H$ is connected. Since all followers are within range of at least one leader, $H$ would be a connected dominating set. 

We now prove that \emph{Backbone-Creation} connects every leader to all leaders within 3 hops of it. In Algorithm \emph{Backbone-Creation}, once leaders gain knowledge of leaders in their vicinity, they will create connections to all such leaders. We will now show that leaders will have knowledge of all leaders within 3 hops of themselves. Note that leaders cannot exist within range of each other, i.e. within 1 hop of each other at the end of \emph{Tree-Cutter}. 

Now, every leader $x$ can hear from nodes within range of it. These nodes are either its followers or other leaders' followers. If a node is a follower of another leader $y$, then a connection can be made to $y$ through one of $y$'s followers. If a node is $x$'s follower, then if there exists another leader $y$ within range of that node, then $y$ can be connected to through one of $x$'s followers. So all leaders within 2 hops of a leader will be connected to it. 

Now consider the case where two leaders are 3 hops apart. Since two leaders $u$ and $v$ cannot be in range of each other, this implies that the nodes comprising the 2 intermediate hops are followers. Note that the followers may not necessarily be those of $u$ and $v$ but of other leaders. If there exists a follower within range of each leader and these followers are within range of each other, then in the first stage of \emph{Backbone-Creation}, these two followers will become aware of each other. In the second stage, the leaders will become aware of each other and the possible connection between them. Since each leader chooses connections so that it is connected to all other leaders known to it, these leaders will become connected to each other. Thus, if two leaders exist within 3 hops of each other, they will be connected to each other. Thus, $H$ is a connected dominating set.

	\item To prove that the degree of any node in $H$ relative to other nodes in $H$ is a constant, let us consider a leader $x$. It will be surrounded by at most some $X$ leaders who are at most 3 hops from it. For each such leader $y$, $x$ specifies one of its followers as a connector to $y$ and $y$ may specify another of $x$'s followers as a connector to $x$ . Thus, there can be at most $2X$ connectors amongst $x$'s followers. Similarly, each of these $X$ leaders can also have $2X$ followers. Since leaders cannot be within range of each other at the end of \emph{Tree-Cutter}, the maximum number of nodes $x$ could be connected to is $2X + 2X^2 = O(X^2)$. Similarly, any connector in $H$ will be surrounded by at most $2X + 2X^2 + X = O(X^2)$ nodes, where the extra $X$ comes from the surrounding $X$ leaders. Since nodes in $H$ are either leaders or connectors, if we show that $X$ is a constant, then the degree of any node in $H$ will be a constant. 
	
	There is at most one leader per grid box of the pivotal grid. There are a constant number of grid boxes within distance $3r$ (maximum distance between two nodes 3 hops away from each other) from a given grid box in the pivotal grid. Thus, $X$ is a constant and the internal degree of nodes in $H$ is a constant.
	\item In the final stage of the algorithm, only nodes who are leaders or connectors transmit their status using a token passing system. By Lemma~\ref{lem:ssf-works}, any node in $H$ will learn about other nodes in $H$ within range of it. Thus, the set of nodes in $H$ induces a subgraph of $G$.
	\end{enumerate}
\item We now need to prove that the number of nodes of $H$ is at most a constant multiple of the minimum sized connected dominating set of $G$. Let us first get an upper bound on the number of nodes in $H$. We already derived an upper bound on the number of neighbors of each leader in $H$ as $2X + 2X^2$. Assuming that there are $L$ leaders in $H$, we can arrive at an upper bound for the number of nodes in $H$ as $L(2X + 2X^2) = O(L X^2) = O(L)$ since $X$ is a constant. We need to show that $L = O(s\_c\_d)$, since then the number of nodes in $H$ will be $O(L) = O(s\_c\_d)$, which is the required property. 

Consider all grid boxes containing nodes of $G$. Since the minimum connected dominating set must span the entire graph, there must exist at least one node in the set per $20$ grid boxes, since the range of a node in one grid box allows it to reach at most $20$ other grid boxes. Hence, if altogether $g$ grid boxes are occupied by nodes in $G$, then at least $\frac{g}{20}$ nodes will be present in the minimum connected dominating set, so $s\_c\_d = \Omega(g)$. Now there exists at most one leader per grid box so $L = O(g) = O(s\_c\_d)$. Therefore, the number of nodes in $H$ is $O(s\_c\_d)$.

\item We need to show that every node in $G \backslash H$ is associated with exactly one entry point in the backbone. In our case, we choose the leaders to act as the unique entry points. All leaders are nodes of $H$ according to Algorithm \emph{Backbone-Creation}. Therefore, the only nodes remaining in $G \backslash H$ are followers. Every follower identifies with exactly one leader and that leader is a part of $H$. Hence, the property is satisfied.

\item Finally, we must prove that the asymptotic diameters of $G$ and $H$ are the same. Every node of $G$ is contained within the backbone or within one hop of a node of the backbone. One property of our model is that any two nodes within range of each other have an edge between them in the communication graph $G$. Since $H$ is an induced subgraph of $G$, the asymptotic diameter of the two graphs must be the same.
\end{enumerate}

Thus, a backbone is created after using \emph{Backbone-Creation}. Since each of the 4 stages consists of token passing by each leader to some subset of its neighbors using $(N,c$)-ssf schedules, the additional running time is $O(\Delta \lg N)$ rounds. This is because each ($N,c$)-ssf schedule takes $O(\lg N)$ rounds and there will be at most $O(\Delta)$ such schedules run since $\Delta$ is the maximum number of nodes surrounding a given node. Adding in the running times of \emph{Tree-Grower} and \emph{Tree-Cutter}, the total running time of Algorithm \emph{Backbone-Creation} is $O(n \lg N)$ rounds.
\qed

\subsection{Communication in Backbone Subnetwork}
Algorithm \emph{Backbone-Message-Exchange} is used to transmit messages between two nodes in the backbone and Algorithm \emph{Backbone-Message-Transmit} is used by a node outside the backbone to transmit a message to its leader. Note that both these procedures are called assuming the context of a network that already executed \emph{Backbone-Creation} and so no parameters are passed as relevant information is assumed to be present at each node.

\begin{theorem}\label{the:back-communicate}
Algorithm \emph{Backbone-Message-Exchange} guarantees the exchange of a message between every pair of connected nodes in the backbone in time $O(\lg N)$ rounds. It also guarantees that every leader in the backbone successfully transmits its message to all its neighbors in $O(\lg N)$ rounds. Algorithm \emph{Backbone-Message-Transmit} guarantees the transmission of a message from a node outside the backbone to its leader in the backbone in time $O(\Delta \lg N)$ rounds.
\end{theorem}

Algorithm \emph{Backbone-Message-Exchange} involves having every leader continuously take part in a token passing routine with all its connectors. Algorithm \emph{Backbone-Message-Transmit} involves every leader taking part in a continuous token passing routine with all of its followers. When a node gets the token, it can transmit a message if it has one.


\alglanguage{pseudocode}
\begin{algorithm}
\caption{Backbone-Message-Exchange, run by each node $u$}
\label{prot:backbone-message-exchange}
\begin{algorithmic}[1]
	\State 	Tree $T$ stores the children of a leader. For followers, $T$ is empty.
	\State $num\_phases$ is no. of phases to run algorithm. If leader, $num\_phases$ is two times no. of children in $T$, else it's 2.
	\State For a node in backbone, tree $H$ stores other nodes in backbone within range. For a node not in the backbone, $H$ is empty.

	\State	
	
	\If {$u$ is a leader}
		\State Add every node that is $u$'s child in both $T$ and $H$ to list of connectors
		\For {Every connector, $v$, in $u$'s list of connectors}
			\State Execute($N,c$)-SSF: Transmit $u$ message, if any, and $num\_phases$
			\State Execute($N,c$)-SSF: Pass token to $v$
			\State Execute($N,c$)-SSF: Store any new messages $u$ hears
			\State Execute($N,c$)-SSF: Do nothing
		\EndFor

	\Else \Comment{$u$ is a connector.}
		\BlockOn {Execute the following $num\_phases$ times}
			\If {$u$ has a token}
				\State Execute($N,c$)-SSF: Transmit $u$'s message, if any. Store any new messages $u$ hears
				\State Execute($N,c$)-SSF: Pass token back to $u$'s leader
			\Else
				\State Execute($N,c$)-SSF: Store any new messages $u$ hears and update $num\_phases$ if needed
				\State Execute($N,c$)-SSF: Listen to see if $u$ receives a token
			\EndIf
		\BlockOff
	\EndIf
\Statex
\end{algorithmic}
\end{algorithm}

\alglanguage{pseudocode}
\begin{algorithm}
\caption{Backbone-Message-Transmit, run by each node $u$}
\label{prot:backbone-message-transmit}
\begin{algorithmic}[1]
	\State 	Tree $T$ stores the children of a leader. For followers, $T$ is empty.
	\State $num\_phases$ is no. of phases to run algorithm. If leader, $num\_phases$ is two times no. of children in $T$, else it's 2.

	\If {$u$ is a leader}
		\For {Every child, $v$, in $T$}
			\State Execute($N,c$)-SSF: Transmit $num\_phases$
			\State Execute($N,c$)-SSF: Pass token to $v$
			\State Execute($N,c$)-SSF: Store any new messages $u$ hears
			\State Execute($N,c$)-SSF: Do nothing
		\EndFor
		
	\Else \Comment{Node is a follower.}
		\BlockOn {Execute the following $num\_phases$ times}
			\If {$u$ has a token}
				\State Execute($N,c$)-SSF: Transmit $u$'s message, if any
				\State Execute($N,c$)-SSF: Pass token back to $u$'s leader
			\Else
				\State Execute($N,c$)-SSF: Listen for new value of $num\_phases$ from leader and update if necessary
				\State Execute($N,c$)-SSF: Listen to see if $u$ receives a token
			\EndIf
		\BlockOff
	\EndIf
\Statex
\end{algorithmic}
\end{algorithm}

\subsubsection{Proof of Theorem~\ref{the:back-communicate}}
We first prove that message transmissions are successful in both algorithms and then go on to bound their running times. Every leader has one token that it passes only to its followers or connectors, i.e. nodes within range $r$. Leaders are the same ones created after using \emph{Tree-Cutter} and by Theorem~\ref{the1}, there is at most one leader per grid box of the pivotal grid. There are 21 grid boxes within range of a node (including the node's own grid box). Thus, at most $21$ tokens can be present in a grid box at any given time and by Theorem~\ref{the:ssf-replace}, using an $(N,c)$-ssf with a sufficiently large $c$ allows all transmitting nodes (those with tokens) to successfully transmit their messages.

Algorithm \emph{Backbone-Message-Exchange} is run for every set of leaders and their connectors. The property of a backbone network is that every node has a constant internal degree. Therefore, a leader passes the token to a constant number of connectors. Furthermore, since we are running ($N,c$)-ssf schedules, each schedule takes $O(\lg N)$ rounds to complete. For two connected nodes in the backbone to exchange messages, each one of them must get a token and transmit. Getting a token will take at most $O(m \lg N)$ time for each of them, where $m$ is the maximum internal degree of any node in $H$. Since $m$ is a constant it takes totally $O(\lg N)$ time to exchange a message between every such pair of nodes. Further, it is clear that in the same number of rounds, each leader in the backbone successfully transmits its message to all its neighbors.

Algorithm \emph{Backbone-Message-Transmit} is run for every set of leaders and their followers. The leader passes a token to each of its followers which then transmits its message according to an ($N,c$)-ssf schedule once it receives the token. Since each leader has at most $\Delta$ followers, where $\Delta$ is the maximum degree of the communication graph, the maximum time it takes for any node to transmit its message is $O(\Delta \lg N)$.
\qed

\section{Multi-Broadcast}\label{sect:multi-broadcast}

The problem of multi-broadcast is to transmit the information held by $k$ nodes, $1 \leq k \leq n$, to all nodes in the network. Each node holds a unique message. Our result is captured by the following theorem.

\begin{theorem}\label{the:multi-broadcast}
Assuming all nodes are initially awake, the Algorithm \emph{Multi-Broadcast} achieves multi-broadcast in $O(n \lg N)$ rounds.
\end{theorem}

This algorithm takes the backbone formed by \emph{Backbone-Creation} and uses it to transmit information through the network. Once the backbone is formed and nodes have waited until $O(n \lg N)$ rounds are over, nodes first run \emph{Backbone-Message-Transmit} and then wait until $O(n \lg N)$ rounds are over. Then nodes subsequently run \emph{Backbone-Message-Exchange}  for $O(n \lg N)$ rounds. The exact values for wait times are calculated based on the maximum possible run time of the preceding process. This is to ensure that all nodes are synchronized at the start of every algorithm they run. Note that we assume a total ordering on all of the $k$ messages in order to uniquely pick one message to transmit from a selection of messages.


\alglanguage{pseudocode}
\begin{algorithm}
\caption{Multi-Broadcast(No. of nodes $n$), run by each node $u$}
\label{prot:multi-broadcast}
\begin{algorithmic}[1]
	\State Backbone-Creation($n$)
	\State Wait until $(9490n + 9471)c_1 \lg N$ rounds completed from start of Backbone-Creation
	\State Backbone-Message-Transmit
	\State Wait until $4nc_1 \lg N$ rounds completed from start of Backbone-Message-Transmit
	\State Perform Backbone-Message-Exchange for $960 n c_1 \lg N$ rounds. Each time $u$ has a token and transmits a message, choose the message with the lowest associated label that has not already been transmitted to transmit.
\Statex
\end{algorithmic}
\end{algorithm}

\subsection{Proof of Theorem~\ref{the:multi-broadcast}}
We argue the correctness of \emph{Multi-Broadcast} as follows. We first argue that after executing \emph{Backbone-Message-Transmit} and waiting the appropriate time, every leader within the backbone has the messages of all nodes outside the backbone that consider it a leader. We consider a single message from one source that needs to be broadcast to all locations and bound the number of times we need to execute \emph{Backbone-Message-Exchange} for it to be received by all nodes in the network. We then bound the additional number of executions of \emph{Backbone-Message-Exchange} it will take for $k \leq n$ messages to be transmitted to all nodes in the network.

First, it is clear from Lemma~\ref{the:back-communicate} that after the execution of \emph{Backbone-Message-Transmit}, all nodes outside the backbone will be able to transmit their messages to their corresponding leaders in the backbone. Thus, every possible message to be transmitted to all nodes is present at least one node within the backbone. 

Now, consider a single message on the backbone that must be transmitted to all nodes in the network. A property of the backbone is that its diameter is asymptotically the same as that of the network. Since nodes do not know the value of the diameter, we upper bound its value by $n$. Thus, after $n$ executions of \emph{Backbone-Message-Transmit}, all nodes in the network know the message. The running time of $n$ execution of \emph{Backbone-Message-Transmit} is upper bounded by $480 n c_1 \lg N$ rounds because for a given leader, the number of connectors (to other leaders) is bounded by $120$ and each node runs $4$ $(N,c)$-ssfs. Thus it takes $480 c_1 \lg N$ rounds to guarantee that \emph{Backbone-Message-Backbone} successfully executed once.

We now consider $k \leq n$ unique messages which must be transmitted to all the nodes in the network. We show that by just adding an additional $480 n c_1 \lg N$ rounds of \emph{Backbone-Message-Transmit}, i.e. guaranteeing the successful completion of another $n$ executions, we can guarantee that all messages reach all nodes. Consider a path of nodes $u, p_1, p_2, \ldots, p_q, v$ between two nodes in the backbone $u$ and $v$ consisting of other nodes within the backbone. We say a message $m_1$ delays another message $m_2$ in path $u, p_1, p_2, \ldots, p_q, v$ if some node in the path has both messages $m_1$ and $m_2$ and needs to transmit $m_1$ before transmitting $m_2$. Once a message $m_1$ delays a message $m_2$ from being transmitted in that path, $m_1$ never subsequently delays $m_2$ again in that path. The only reason $m_2$ might have to wait at another node for $m_1$ to be transmitted is because there existed another message $m_3$ with a lower associated label than $m_1$ or $m_2$. In this case, the delay to $m_2$ is attributed to $m_3$. For any given message, there are at most $k-1 < n$ other messages that can delay it. Thus an additional delay of at most $480 n c_1 \lg N$ rounds may be imposed on a message as it travels through the network. 

Thus after $960 n c_1 \lg N$ rounds of executing \emph{Backbone-Message-Exchange} all nodes within the backbone have transmitted all messages. Furthermore, by Lemma~\ref{the:back-communicate}, all nodes outside the backbone have received all messages as well. Finally, it is clear that the running time of \emph{Multi-Broadcast} is $O(n \lg N)$ rounds.
\qed

\section{Conclusions}\label{sect:conclusions}

  In this work, we show several applications of the technique known as SSF Based Dilution. We use it to provide the first deterministic algorithm for multi-broadcast from uncoordinated wakeup when knowledge of nodes' physical coordinates and neighborhoods is not known. Additionally, we use SSF Based Dilution to construct algorithms to deterministically solve multi-broadcast and backbone creation for spontaneous wakeup. We present an open problem we feel is of interest.\\
  \textbf{Open Problem:} \emph{Tree-Grower} requires a message size of $O(\Delta \lg N)$ bits. An alternate algorithm with a smaller requirement would be advantageous.

\section*{Acknowledgements}
We are very grateful to Darek Kowalski for several enriching discussions at various stages of this work. The first author would also like to thank the members of the theory group of the Dept. of Computer Science \& Engineering at IIT Madras for their input and feedback when this material was presented as a talk. We are also very grateful to the anonymous reviewer for helping to improve the presentation, style, and some technical matters of this paper. This research did not receive any specific grant from funding agencies in the public, commercial, or not-for-profit sectors. 

\bibliographystyle{abbrv}
\bibliography{reference}

\end{document}